\documentclass[aps, pra,twocolumn, showpacs,superscriptaddress,10pt,floatfix]{revtex4-1}

\usepackage[utf8]{inputenc}
\usepackage[T1]{fontenc}
\usepackage{verbatim}
\usepackage{bm}
\usepackage{graphicx}
\usepackage{grffile}
\usepackage{amsthm}
\usepackage{epsfig}
\usepackage{amsmath,amssymb,amsfonts,mathtools,bbm,MnSymbol,mathrsfs,xfrac}
\usepackage{longtable}
\usepackage{tabularx}
\usepackage{colortbl}
\usepackage[table]{xcolor}
\usepackage[breaklinks=true,colorlinks=true,linkcolor=blue,urlcolor=blue,citecolor=blue]{hyperref}
\usepackage{comment}
\usepackage{color}
\usepackage[makeroom]{cancel}
\usepackage{tikz}
\usetikzlibrary{decorations.pathreplacing,calligraphy}
\usepackage{graphicx}
\DeclareGraphicsExtensions{.pdf,.png,.jpg}
\usepackage{algorithm2e}
\usepackage{ dsfont }

\usepackage{soul}
\pdfoutput=1

\renewcommand{\[}{\begin{equation}}
\renewcommand{\]}{\end{equation}}

\newcommand{\ket}[1]{|#1\rangle}
\newcommand{\bra}[1]{\langle#1|}

\newcommand{\pro}[2]{|#1\rangle\langle#2|}
\newcommand{\mean}[1]{\langle#1\rangle}
\newcommand{\abs}[1]{|#1|}

\newcommand{\tr}{\mathrm{tr}}

\newcommand{\norm}[1]{\left\|#1\right\|}
\newcommand{\normt}[1]{\|#1\|}

\newcommand{\R}{{\hat{\rho}}}
\newcommand{\A}{{\alpha}}
\newcommand{\G}{{\gamma}}

\newcommand{\bi}{{\boldsymbol{i}}}
\newcommand{\bj}{{\boldsymbol{j}}}
\renewcommand{\bm}{{\boldsymbol{m}}}

\newcommand{\ham}{{\hat{H}}}

\definecolor{mygray}{gray}{0.6}

\theoremstyle{definition}

\newtheorem{theorem}{Theorem}
\newtheorem{lemma}{Lemma}
\newtheorem{corollary}{Corollary}

\definecolor{dfcol}{cmyk}{1, 0.2108, 0.13, 0.3}
\newcommand{\df}[1]{\ifthenelse{\boolean{}}{\textcolor{dfcol}{[{\bf DF}: #1]}}{}}

\newcommand{\hamz}{\hat{H}}
\newcommand{\hamd}{\hat{V}}
\newcommand{\hamh}[1]{\hat{h}(#1)}
\newcommand{\hamv}[1]{\hat{v}(#1)}
\newcommand{\floor}[1]{\lfloor #1 \rfloor}
\newcommand{\B}{\beta}
\newcommand{\haml}{\hat{H}^{(l)}}

\renewcommand{\L}{\lambda}

\begin{document}

% title
\title{Quantum Charging Advantage Cannot Be Extensive Without Global Operations}

\author{Ju-Yeon Gyhm}
\email[E-mail: ]{kjy3665@snu.ac.kr}
\affiliation{Center for Theoretical Physics of Complex Systems, Institute for Basic Science (IBS), Daejeon 34126, Republic of Korea}
\affiliation{Department of Physics and Astronomy, Seoul National University, 1 Gwanak-ro, Seoul 08826, Korea}

\author{Dominik \v{S}afr\'{a}nek}
\email[E-mail: ]{dsafranekibs@gmail.com}
%\thanks{\\DR and D\v{S} contributed equally to this work}
\affiliation{Center for Theoretical Physics of Complex Systems, Institute for Basic Science (IBS), Daejeon 34126, Republic of Korea}

\author{Dario Rosa}
\email[E-mail: ]{dario\_rosa@ibs.re.kr}
\thanks{\\ D\v{S} and DR contributed equally}
\affiliation{Center for Theoretical Physics of Complex Systems, Institute for Basic Science (IBS), Daejeon 34126, Republic of Korea}

% date
\date{\today}

% abstract
\begin{abstract}
% quantum batteries are important and will revolutionize the world
% there are bounds, but these do not work
% we prove the bound (proving the conjecture)
% we prove that to achieve this sclaing, we need global operations
% but global operations themselves do not guarantee
% this bound is in some sense achievable
%this adds to the literature of maximally quadratic improvement over classical methods
% this bound can be used for general hamiltonian instead of just hamiltonians descibing $L$ identical cells
Quantum batteries are devices made from quantum states, which store and release energy in a fast and efficient manner, thus offering numerous possibilities in future technological applications. They offer a significant charging speedup when compared to classical batteries, due to the possibility of using entangling charging operations. We show that the maximal speedup that can be achieved is extensive in the number of cells, thus offering at most quadratic scaling in the charging power over the classically achievable linear scaling. To reach such a scaling, a global charging protocol, charging all the cells collectively, needs to be employed. This concludes the quest on the limits of charging power of quantum batteries and adds to other results in which quantum methods are known to provide at most quadratic scaling over their classical counterparts.
\end{abstract}

% pacs and keys

\maketitle
%background
%motivation for doing all this / how does it affect society as a whole and why is it interesting to a general reader
%problem (this could include the discussion about the charging advantage)
%solution (contents of paper)

\emph{Introduction.---}
%quantum technologies are awesome nad real
In recent years tremendous efforts have been devoted to developing quantum technologies, which are now coming to fruition in several fields of practical use.
Among the largest successes is quantum metrology~\cite{giovannetti2011advances}, which led to the detection of gravitational waves~\cite{abbott2016observation}, quantum cryptography~\cite{pirandola2020advances}, which finds applications in communicating sensitive data~\cite{elliott2005current,chen2021integrated}, quantum computing, which promises to revolutionize chemistry~\cite{cao2018quantum} as well as to speed up or solve important problems in optimization, cybersecurity and data analysis~\cite{aaronson2008limits}, and nanoscale thermodynamic devices, which offer unprecedented precision in thermometry~\cite{menges2016temperature}. At large, society is moving toward quantum technologies, because they promise to offer faster, smaller, and more precise devices.

%quantum batteries are needed and will revolutionize everything
All of these achievements require an efficient way of storing and using energy, as well as fast charging and discharging. The necessity of charging and discharging goes well beyond the quantum world. Examples are electric vehicles where the charging time is one of the main bottlenecks in preventing the widespread use of such technology, or future fusion power plants, in which a large amount of energy needs to be pumped in a short amount of time and discharged in an instant to start the reaction. In the quantum world, nanoscale devices will require nanoscale batteries, with no energy to spare.

%All these and all the sucesses of the quantum technologies led people to consider possibility of using quantum effects to achieve all of those desiderata
%Quantum batteries promise to provide all of these advantages, in principle.

Outstanding successes of quantum technologies prompt a question whether quantum effects can also improve the energy storage to satisfy current and future demands. This leads to the notion of quantum battery, which is a quantum mechanical system acting as an energy storage, and in which quantum effects are expected to provide significant advantages over its classical counterpart (see \cite{campaioli2018quantum, bhattacharjee2020quantum} for reviews). Starting from the work of Alicki and Fannes~\cite{PhysRevE.87.042123}, the possibility of using quantum effects (like coherence and entanglement) to increase the performance of a quantum battery has been heavily studied.
These studies address several figures of merit, such as work extraction \cite{PhysRevE.87.042123, Hovhannisyan_2013}, energy storage \cite{PhysRevB.98.205423, PhysRevE.99.052106, PhysRevResearch.2.023095, Quach_2020, Crescente_2020fluct}, charging stability \cite{Friis_2018, PhysRevB.100.115142, PhysRevE.100.032107, Rosa_2020}, available energy \cite{PhysRevLett.122.047702, PhysRevLett.122.210601, PhysRevResearch.2.033413} (with the notion of ergotropy \cite{Allahverdyan_2004}) and charging power~\cite{Binder_2015, Campaioli_2017, Le_2018, PhysRevLett.120.117702, PhysRevB.99.205437, Crescente_2020, PhysRevA.101.032115, Rossini_2020, zakavati2020bounds, ghosh2020fast, seah2021quantum, PhysRevLett.125.040601}, the last being the actual focus of this letter.
It has been shown~\cite{Campaioli_2017} that quantum effects lead to a speedup in the charging power of a quantum battery. The source of this quantum speedup lies in the use of entangling operations, in which the cells are charged collectively as a whole. Those operations, where the number of cells that are being entangled together collectively scales with the system size (i.e., creating multi-partite entanglement), are called \emph{global operations}. In contrast, classical batteries are charged in parallel, meaning that each cell is charged independently of each other. The advantage of this collective versus parallel charging is measured by the ratio $\Gamma$, called the \emph{quantum charging advantage}~\cite{Campaioli_2017}. However, it is still not known \emph{how large} the quantum advantage is in general. To this end, the best known result is~\cite{Campaioli_2017}
%Quantum battery uses the entangling operations which, contrary to classical operations in which each cell is separately charged,  provide a significant quantum speed up. However, it is still not known \emph{how much} advantage a quantum battery can offer when compared to a classical battery.
\[
\label{eq:quantum_advantage_old}
\Gamma < \gamma \big(k^2 (m - 1) + k \big),
\]
in which $\gamma$ is a model-dependent constant, $k$ is the maximum number of cells that are collectively charged, while $m$ (called \emph{participation number}) is the maximum number of parallel charging operations in which a single cell appears.

In principle, this bound allows for a super-extensive scaling of the quantum advantage, meaning that the advantage can scale more than linearly with the number $L$ of cells. For example, consider a charging protocol that has a finite and fixed value of $k$ but having \textit{all-to-all} couplings. In such a case, the participation number of a given cell is of order $m =
{{L-1}\choose{k-1}}\approx (L-1)^{k-1}/(k-1)!$, leading to a quantum advantage of order $L^{k-1}$.

%However, there is a numerical evidence~\cite{} that this advantage is much smaller --- at most extensive --- of order $L$.

This prompted a race toward finding the best possible scaling --- the authors of~\cite{Campaioli_2017} found that the scaling is of order $L$ at most through an extensive numerical search, and proposed a conjecture that this extensive scaling cannot be surpassed. The search for scaling advantages continued in 
% This started an extensive race towards finding the best possible scaling in this advantage
% In the past there has been an extensive race towards finding the best possible scaling in this advantage. 
~\cite{Le_2018, PhysRevLett.120.117702, PhysRevB.98.205423, PhysRevE.99.052106, zhang2018enhanced, PhysRevA.101.032115}, which also showed at most extensive scaling, but it was later shown that some of these advantages were not caused by genuine quantum effects~\cite{Juli_Farr__2020}. A genuine, extensive, quantum advantage was found in~\cite{Rossini_2020}, in a setup including \textit{both} global charging operations and all-to-all couplings. The conjecture still held, but remained unproven, together with uncertain role as to which all-to-all interactions play in determining the quantum advantage.

In this letter, we prove this conjecture, showing that a quantum battery provides \textit{at most} extensive advantage over classical batteries. Furthermore, we show that this scaling is achievable only via \textit{global} charging operations, \textit{i.e.}, we show that all-to-all interactions, and more generally, the participation number, does not provide any scaling advantage. 

We first provide a \textit{general} bound (Theorem \ref{theorem:1}), constraining the maximum charging power achievable with a \textit{general} quantum battery with any general Hamiltonian, not necessarily realized by $L$ identical cells, thus including also more general cases described in the literature \cite{Le_2018}.
%As a consequence (Corollary~\ref{corollary:1.1}), we obtain the bound for the case of a battery made of identical cells, thus proving the aforementioned conjecture. Together with examples showing extensive advantage already found and discussed in the literature, our results concludes the quest for the best possible scaling which can be obtained by quantum batteries.
The conjecture is proven as a consequence of this theorem (Corollary~\ref{corollary:1.1}), applied to the battery made of identical cells. Together with examples showing extensive advantage \cite{Binder_2015, Rossini_2020}, already found and discussed in the literature, this result concludes the quest for the best possible scaling which can be obtained by quantum batteries.

\emph{Setup.---}
We consider quantum batteries made out of a time independent initial Hamiltonian, $\hamz$, having discrete spectrum. At time $t = 0$ a possibly time-dependent driving Hamiltonian, $\hamd (t)$, is turned on and the initial state $\R_0$ is evolved according to the quench
\[\label{eq:state_evolution}
\tfrac{d\R_t}{dt}=-i[\hamd(t),\R_t] .
%\rho(t) = \hat{U}_t \R \hat{U}_t^\dagger \ ,
\]
%where we denoted with $U(t)$ the time evolving operator, $ \hat{U}(t)=\mathcal{T}e^{-i\int_{0}^{t}{\hamz+\hamd(s)ds}}$.
The energy stored in the battery, measured by the initial Hamiltonian, changes from $E(0) = \tr(\hamz \R_0)$ to $E(t) = \tr(\hamz \R_t)$ during time evolution.
Charging the battery means reaching large values of $E(t) - E(0)$.

An important figure of merit is the instantaneous charging power of the battery.
It is defined as the instantaneous change in the energy stored per unit of time:
\[\label{eq:power}
P(t)=\tr\big(\hamz\,  \tfrac{d\R_t}{dt}\big). %\big \Big, \bigg, \Bigg,
\]
where we used that $\hamz$ is time-independent.

Generally, the instantaneous power is bounded in terms of the commutator between $\hamz$ and the driving term
\[\label{eq:most_general}
\abs{P(t)}\leq \norm{[\hamz,\hamd(t)]} \leq 2\norm{\hamz}\norm{\hamd(t)}
\]
through the operator norm~\footnote{It is proved by inserting Eq.~\eqref{eq:state_evolution} to Eq.~\eqref{eq:power}, which gives $P(t)=\tr(\R[\hamz,\hamd])=\sum_k\rho_k\bra{\psi_k}[\hamz,\hamd]\ket{\psi_k}\leq \norm{[\hamz,\hamd]}\leq 2\norm{\hamz}\norm{\hamd}$, using the spectral decomposition $\R=\sum_k\rho_k\pro{\psi_k}{\psi_k}$ and $\sum_k\rho_k=1$. See, for example, Ref.~\cite{Juli_Farr__2020}}.

However, the driving is often limited in realistic situations. For example, in lattice systems, the interaction couples only nearby sites and therefore $\hamd$ transfers energy only between not-too-distant energy levels of the initial Hamiltonian $\hamz$.
Taking the spectral decomposition of the initial Hamiltonian to be $\hamz = \sum_{j} E_j \ket{E_j}\bra{E_j}$, where we assume the energy levels $E_j$ being ordered, we express the driving Hamiltonian as $\hamd=\sum_{j , m=1}^{N} V_{jm}\ket{E_{j}}\bra{E_m}$. The limiting property is formalized as follows: we define $\Delta E$ as the minimum value such that for all $j$ and $m$,
%, \textit{i.e.} $E_j \leq E_m$ when $j < m$ 
% we express the driving Hamiltonian in the basis of the initial Hamiltonian as $\hamd=\sum_{j , m=1}^{N} V_{jm}\ket{E_{j}}\bra{E_m}$. The limiting property is formalized as follows: we define $\Delta E$ as the minimum value such that
% \[\label{eq:property}
% V_{jm} = 0\quad \text{when}\quad \abs{E_{j}-E_{m}} >\Delta E.
% \]
\[\label{eq:property}
\mathrm{when}\quad | E_{j}-E_{m} | >\Delta E,\quad \mathrm{then}\quad V_{jm} = 0.
\]

% For example, in most cases the driving Hamiltonian $\hamd$ transfers energy only between not-too-distant energy levels of the initial Hamiltonian $\hamz$. 
% This is the case occurring in the most common systems, such as lattices, in which the interaction Hamiltonian, which defines our driving, affects only nearby sites.

%In many physical systems, however, the driving Hamiltonian $\hamd$ is often constrained. For example, it is assumed that $\hamd$ only transfers energy between not-too-distant energy levels of the initial Hamiltonian $\hamz$. This is the case occurring in the most common systems, such as lattices, which often constrain the interaction Hamiltonian to affect only nearby sites, thus naturally leading to this general property.

% There exists 
% a $\Delta E$ such that
% \[\label{eq:Vdefinition}
% \text{if} j,m \text{\ is\ such\ that\ } \abs{E_{j}-E_{m}} >\Delta E, \text{then} 
% V_{jm} = 0 
% %V_{jm} = 0 \quad \mathrm{when} \quad  | E_{j}-E_{m} | >\Delta E, 
% \]
% for all $j$ and $m$ that satisfy
% %The system is initially prepared in a given state, denoted by $\R(0)$.

Thus, it is natural to look for a more precise bound than eq.~\eqref{eq:most_general}, taking this common property into account.

%It is worth to emphasize that a bound like eq.~\eqref{eq:general_inequality_of_pewer} is completely \textit{state-independent} since it relies on  the spectral properties of $\hamz$ and $\hamd$ only.

%\section{Main result} 
\emph{Main result.---} In the conditions of eq. \eqref{eq:property},  we now show that a more stringent bound can be derived.

\begin{theorem}\label{theorem:1}
%The amount by which the energy can increase has been of great scientific interest~\Ref{}. The following theorem gives more accurate upper bound of power of batteries than eq.~\eqref{eq:general_inequality_of_pewer} under the constraint $\Cz$. To instruct the theorem, let us consider the
%initial Hamiltonian expressed as spectral decomposition $\hamz=\sum_{j=1}^{N} E_{j}\ket{E_{j}}\bra{E_j}$, where the energies $E_j$ are ordered, \textit{i.e.} $E_j \leq E_m$ for $j < m$.
%We will write the driving Hamiltonian in basis of the initial Hamiltonian as 
For driving that couples energy levels with at most $\Delta E$ energy difference, as expressed by eq.~\eqref{eq:property}, the instantaneous charging power is bounded as
% \begin{equation}\label{eq:main_teorem}
% \abs{P(t)}\leq2\Delta E\norm{\hamd(t)}.
% \end{equation}
\begin{equation}\label{eq:main_teorem}
\abs{P(t)}\leq \Delta E\norm{\hamd(t)- v_{\min}(t)}/2,
\end{equation}
where $v_{\min}(t)$ is the smallest eigenvalue of $\hamd(t)$, and $\norm{~}$ denotes the operator norm.
\end{theorem}

Hence, the operator norm of the initial Hamiltonian, central in inequality~\eqref{eq:most_general}, \textit{is not} the relevant figure of merit. Instead, the crucial quantity is the maximal value of energy (as measured by $\hamz$) that can be transferred by $\hamd$ in a single time step. While non-trivial to prove, this result is very intuitive. 
%The charging power is the amount of the energy stored in the battery in a single time step. Thus, it is natural to expect that this change in energy must be bounded by the maximum amount of energy that the driving term can transfer to the system during that step.
The charging power is the amount of the energy stored in the battery in a single time step. Thus, this change in energy must be bounded by the maximum amount of energy that the driving term can transfer to the system during that time.

%In turns, this is not, in general, given by the operator norm of $\hamz$ but it is given by the maximum energy jump that $\hamd$ can achieve. This latter quantity is precisely the definition of $\Delta E$ as in eq.~\eqref{eq:property}.

The fact that the bound is not given by the operator norm of $\hamz$ has a far reaching consequence. As outlined in the Introduction, it has been a matter of active research which combination of initial and driving Hamiltonians can reach a charging power scaling with $\norm{\hamz}$. Theorem \ref{theorem:1} shows that to reach this scaling one needs to consider driving terms having non vanishing matrix elements between the ground state and the highly excited states of $\hamz$.
This latter property defines the so-called \textit{global} charging operations which we discuss more extensively later. Another point of this bound is that it applies to any Hamiltonian $\hamz$, even with interacting cells.

%This way of reasoning also shows that to have a charging power which scales with $\norm{\hamz}$, one needs to consider driving terms having non vanishing matrix elements between the ground state and the highly excited states of $\hamz$.
%This latter property defines the so-called \textit{global} charging operations which we will discuss more extensively later.

%\wl Intuitively, this is a very natural result, since the charging power refer to a very short time step, therefore if not too much energy can be exchanged, this will certainly limit the change in energy.

%This means that the maximum charging power is bounded by the the maximum difference in energy between two eigenstates of initial Hamiltonian which can be connected by a matrix element of driving Hamiltonian.

%That give the maximum power of batteries depend on $\Delta E$, which we call maximum switch, representing the maximum difference in energy between two eigenstates of initial Hamiltonian which can be connected by a matrix element of driving Hamiltonian.

%\section{sketch of the proof}
%\emph{sketch of the proof}.---
\emph{Sketch of the proof.---}
The full proof of the theorem is technically involved. Here we sketch the main idea, while we refer the reader to the Supplementary material for details.

We express the commutator between $\hamz$ and $\hamd$ as an integral of commutators which are more easily and directly bounded. In particular, we define certain operator functions $\hamh{e}$ and $\hamv{e}$, depending on a continuous parameter $e$ and satisfying
\[\label{eq:commutator_as_integral}
\left[\hamz,\hamd\right]=\int_{e=0}^{\Delta E} [\hamh{e},\hamv{e}]\,  de,
\]
as well as the following properties
\[
\norm{\hamh{e}}=\frac{1}{2},\quad\quad\norm{\hamv{e}}=\norm{\hamd}
.
\]%\label{eq:statement2} \label{eq:statement3}

%From these two properties, proving the theorem is a straightforward application of the triangle inequality to eq.~\eqref{eq:commutator_as_integral}.

We apply triangle inequality to eq.~\eqref{eq:commutator_as_integral} to derive bound $\norm{\left[\hamz,\hamd\right]}\leq \Delta E\norm{\hamd(t)}$. 
Since any number $\lambda$ commutes with $\hamz$, and $\hamd'_\lambda=\hamd-\lambda$ also satisfies eq.~\eqref{eq:property}, we can make this bound tighter by minimizing over $\lambda$,
\[
\norm{\left[\hamz,\hamd\right]}\!=\!\inf_\lambda\norm{\left[\hamz,\hamd'_\lambda\right]}\!\leq\! \inf_\lambda\Delta E\!\norm{\hamd\!-\!\lambda}\!=\!\frac{\Delta E}{2}\!\norm{\hamd\!\!-\!v_{\min}}.
\]
The theorem then follows from eq.~\eqref{eq:most_general}.

% Commutator between , which are defined as function of variable $e$ by $\hamz$ and $\hamd$. Because of the property of driving Hamiltonian, the commutator $[\hamz,\hamd]$ expressed in the basis of initial Hamiltonian as
% \[
% \begin{split}\label{eq:commuteHV}
% \left[\hamz,\hamd\right]&=\sum_{j,m=0}^{N}{[\hamz,V_{jm} \ket{E_{j}}\bra{E_{m}}]}\\
% &=\sum_{j,m=0}^{N}{(E_{j}-E_{m})V_{jm} \ket{E_{j}}\bra{E_{m}}},
% \end{split}
% \]
% has a lot of zero elements. We use this property to reduce the number of operator $\hamh{e}$ and $\hamv{e}$, which produce the commutator $\left[\hamz,\hamd\right]$ as summation. Moreover, we manipulate $\hamv{e}$ to it has same operator norm with $\hamd$. 
% Then, $\hamh{e}$ and $\hamv{e}$ satisfy following three statements
% \begin{eqnarray}
% \left[\hamz,\hamd\right]&=&\int_{e=0}^{\Delta E}{[\hamh{e},\hamv{e}]de}.\label{eq:statement1},\\
% \norm{\hamv{e}}&=&\norm{\hamd}\label{eq:statement2},\\
% \norm{\hamh{e}}&=&1\label{eq:statement3}.
% \end{eqnarray}
% The definition of operator $\hamh{e}$ and $\hamv{e}$, and proof of statements is on Appendix~\ref{}. By the statements (eqs.~\eqref{eq:statement1}, \eqref{eq:statement2} and \eqref{eq:statement3}) and triangular inequality of operator norm, we obtain the inequality (eq.~\eqref{eq:main_teorem}) form eq.~\eqref{eq:general_inequality_of_pewer}

%\section{Lattice case}
\emph{Lattice case.---}
%We obtain a corollary about interacting lattice batteries form theorem~\eqref{theorem:1}. 
In case of a battery made by cells, each of them given, for example, by a qubit, Theorem~\ref{theorem:1} provides a much more stringent bound than other known bounds \cite{Campaioli_2017}. 

%Let us consider a battery composed of $L$ identical cells, which has the initial Hamiltonian $\hamz=\sum_{l=1}^{L}\haml$, where $\haml$ is initial Hamiltonian of single cell. 
We consider a battery composed of $L$ identical cells, having initial Hamiltonian 
\[\label{eq:initial_lattice_Hamiltonian}
\hamz=\sum_{l=1}^{L}\haml,
%\hamz=H^{(1)}\otimes\hat{I}\otimes\cdots\otimes\hat{I}\ +\ \cdots\ +\ \hat{I}\otimes\cdots\otimes\hat{I}\otimes H^{(1)}
\]
where $\haml=\hat{I}\otimes\cdots\otimes\hat{I}\otimes \hat{H}_{\mathrm{s}} \otimes\hat{I}\otimes\cdots\otimes\hat{I}$ and $\hat{H}_{\mathrm{s}}$ is the single site Hamiltonian at the $l$-th place. We charge this battery by turning on the driving Hamiltonian,
%The most general form for a driving Hamiltonian of a local interaction quantum battery is as follows
\begin{eqnarray}\label{eq:interaction_Hamiltonian}
\hamd(t)=\sum_{\bi\in K(L,k)}{\hat{V}_{\bi}}(t),
\end{eqnarray}
where, by definition, each term in the summation couples together at most $k$ cells. Expressed mathematically,
\begin{align}
K(L,k)&=\bigcup_{n=1}^k C(L,n),\\
C(L,n)&=\{(i_1,\dots,i_n)| i_1<\cdots<i_n \text{\ and\ }i_j\in \{1,\dots,L\}\},\nonumber
\end{align}
where $C(L,n)$ is a set of all combinations of $n$ sites, %and $[\haml,\hamd_\bi(t)]=0$ when $l\not\in\bi=(i_1,\dots,i_n)$. The corollary follows.
and $V_\bi$ acts as an identity on the site which does not appear in the index, i.e., for any local matrix $\hat{M}^{(l)}=\hat{I}\otimes\cdots\otimes\hat{I}\otimes\hat{M}\otimes\hat{I}\otimes\cdots\otimes\hat{I}$, where $\hat{M}$ is at the $l$-th place, if $l\not\in\bi=(i_1,\dots,i_n)$, then $
[\hat{M}^{(l)},\hamd_\bi(t)]=0$. The cases with $k\propto L$ are called \emph{global operations}. The corollary follows.
\begin{figure}[t!]
\begin{center}
    \includegraphics[width=.9\hsize]{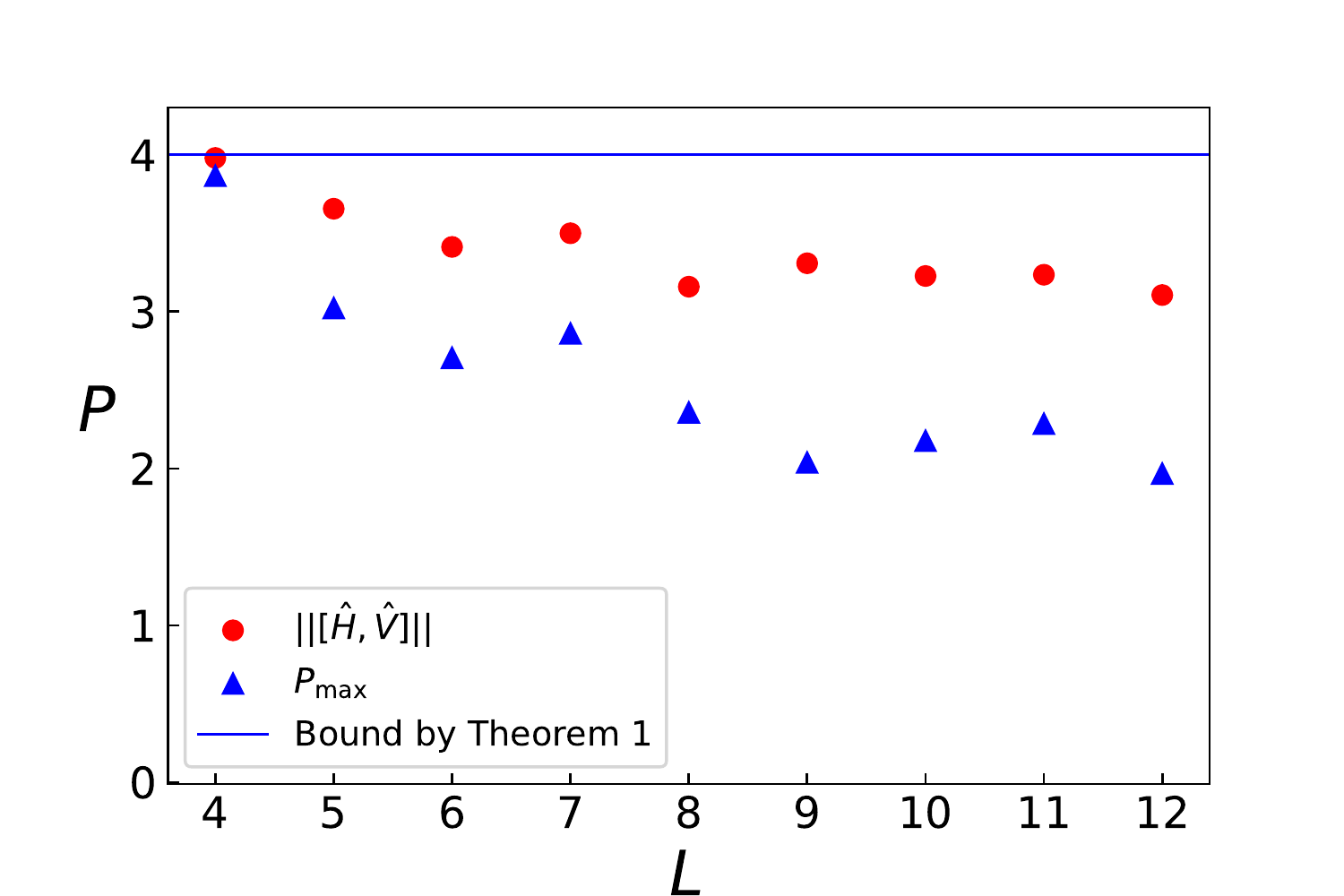}\\
    \ \ (a)\\
    \includegraphics[width=.9\hsize]{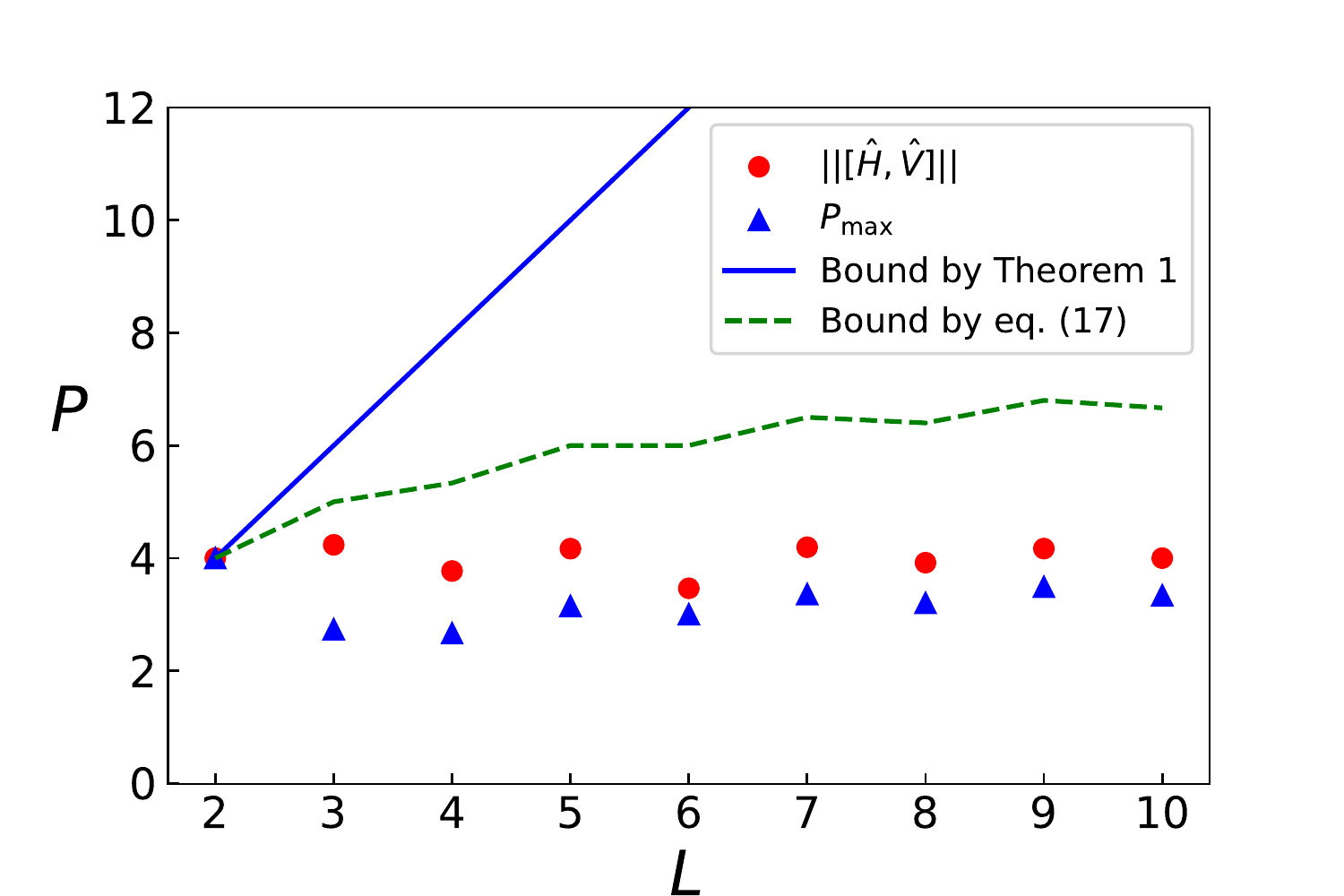}\\
    \ \ (b)
    \caption{(a) Maximum power $P_{\max}$ and maximum operator norm of the commutator $\norm{[\hamz, \, \hamd]}$ as a function of $L$ (maximized over all time and 500 realizations of disorder for each value of $L$), starting from the ground state of initial Hamiltonian $\hamz = \sum_{l=1}^{L}h\hat \sigma^z_l$ and charged by driving Hamiltonian as in eq.~\eqref{eq:driving_Hamiltonian_random}. For this example, we fixed $h =1$. We observe a slight decrease in $P_{\max}$ and $\normt{[\hamz,\hamd]}$ with growing $L$, which is a result of increasing dimensionality of the system, resulting in the lower chance of the of the initial state to be optimal. While keeping the number of realizations of different driving fixed at 500, which means that upper the bound is more difficult to reach. (b) The same as (a) but for a driving given by eq.~\eqref{eq:driving_Hamiltonian_2}. For this example, we additionally fixed $V = 1$.}
\label{fig:figure1}
\end{center}
\end{figure}
\begin{corollary}\label{corollary:1.1}
For initial and driving Hamiltonians~\eqref{eq:initial_lattice_Hamiltonian} and~\eqref{eq:interaction_Hamiltonian}, the instantaneous charging power is bounded as,
\begin{equation}\label{eq:corollary1.1}
\abs{P(t)}\leq k\norm{\hat{H}_{\mathrm{s}}-E_{\mathrm{s}\min}}\norm{\hamd(t)-v_{\min}(t)}/2.
\end{equation}
where $E_{\mathrm{s}\min}$ is the single cell ground state energy.
\end{corollary}
The result is proven by showing that the maximum energy jump $\Delta E$ in this case is given by $k\norm{\hat{H}_\mathrm{s}}$.
As shown in the Supplementary Material, the corollary then follows directly from Theorem~\ref{theorem:1}.
%In the case of non-identical cells, we have $\Delta E=\sum_l\norm{\haml-E_0^{(l)}}$. Details can be found in the supplementary material.
%The corollary is proved by showing $\hamd$ has $\Delta E$ as $k\norm{\haml}$.

%The consequences of the bound in eq.~\eqref{eq:corollary1.1} are remarkable. In particular, it rules out the possibility of having extensive quantum charging advantage without global charging operations. Let us assume that the integer $k$, controlling the maximum number of cells coupled together by $\hamd$, \textit{does not} scale with the lattice size, $L$.
%Under such an assumption, it follows that the quantum advantage, which reads $\Gamma \leq \gamma k$, cannot scales with $L$, as extensively foreshadowed in the Introduction.

The consequences of the bound in eq.~\eqref{eq:corollary1.1} are remarkable. In particular, it rules out the possibility of having extensive quantum charging advantage without global charging operations. 

To show that, we need to discuss each of these terms separately: as explained above, $k$ is the number of cells being coupled together by the driving, and thus $k\leq L$. $\normt{\hat{H}_\mathrm{s}-E_{\mathrm{s}\min}}$ is a number that depends on particulars of a single cell but does not scale with $L$. The last term, $\norm{\hamd(t)-v_{\min}(t)}$, which we call \emph{potential} (in analogy with electric circuits), can be in principle made arbitrarily large. Physically, this would correspond to investing a very large/infinite energy into the driving. With larger driving energy, the charging is faster. Thus, to compute the quantum charging advantage, we need to compare the parallel and quantum scaling on an equal footing, by assuming that the energy scale that is invested into the driving is the same in both cases. 
%This is done by fixing the potential $\norm{\hamd-v_{\min}}=\norm{\hamd^\parallel-v_{\min}^\parallel}$ to be the same for the quantum and parallel driving.
We do that by fixing the potential of the quantum driving to be at most equal to the potential of the parallel driving, $\normt{\hamd-v_{\min}}\leq \normt{\hamd^\parallel-v_{\min}^\parallel}$. This is the constraint C0, introduced and argued for in Ref.~\cite{Campaioli_2017}. 

Parallel charging is given by $k=1$ in driving Hamiltonian~\eqref{eq:interaction_Hamiltonian}, while the initial state is assumed to be a product state $\R=\R_{\mathrm{s}}^{\otimes L}$. Thus, in the parallel charging scenario the driving affects each cell independently. From this, we easily calculate that both the potential $\normt{\hamd^\parallel-v_{\min}^\parallel}=L\normt{\hamd_{\mathrm{s}}^\parallel-v_{{\mathrm{s}}\min}^\parallel}$ and the charging power $P^\parallel=LP_{\mathrm{s}}^\parallel$ scale linearly with $L$. $\R_{\mathrm{s}}$, $\normt{\hamd_{\mathrm{s}}^\parallel-v_{{\mathrm{s}}\min}^\parallel}$, and $P_{\mathrm{s}}^\parallel$ denote the state, potential, and charging power of a single cell, respectively.

Combining eq.~\eqref{eq:corollary1.1}, constraint C0, and the results for parallel charging, we bound the quantum advantage as
\[
|\Gamma|=\frac{\left|P\right|}{\left|P^\parallel\right|}\leq \frac{k\normt{\hat{H}_{\mathrm{s}}-E_{\mathrm{s}\min}} \normt{\hamd_{\mathrm{s}}^\parallel-v_{{\mathrm{s}}\min}^\parallel} L}{2\abs{P_{\mathrm{s}}^\parallel} L}=\gamma k,
\]
where $\gamma$ is $L$ and $k$-independent. Thus, the quantum advantage scales with the maximum number $k$ of cells that are coupled together by $\hamd$. If this number \textit{does not} scale with the lattice size, $L$, then the quantum advantage cannot scale with $L$, as extensively foreshadowed in the Introduction. The extensive scaling is possible only for global interactions, $k\propto L$. By showing that the only source of the quantum advantage comes from the global entangling operations, we showed that the advantage comes from genuine quantum effects. This addresses the discussion of the role of quantum-ness posed in relation to the bound on charging power found in~\cite{Juli_Farr__2020}.

Finally, we ask what is the maximal scaling of power with $L$ that a quantum charging protocol can achieve. Clearly, using Eq.~\eqref{eq:corollary1.1}, the maximum charging power is given by the product of $k$, and by whatever scaling can be constructed from $\normt{\hamd-v_{\min}}$ (for now leaving constraint $C0$ behind). It is possible to artificially construct some driving Hamiltonians that scale super-extensively, i.e., with higher powers of $L$~\cite{Le_2018,PhysRevLett.120.117702}. However, such models are unphysical~\cite{Juli_Farr__2020,Rossini_2020}, because they would lead to a free energy that is super-extensive in the thermodynamic limit. Therefore, for any physical model, considering extensive energy $\normt{\hamd-v_{\min}}\sim L$, the maximal charging power scales at most quadratically $P\sim L^2$, for global operations $k\propto L$. (Compare with the linear scaling of parallel charging.) 

As an illustrative example, we study charging of a quantum battery by means of a driving Hamiltonian obtained via a simple generalization of the celebrated SY Hamiltonian \cite{PhysRevLett.70.3339}, \textit{i.e.} a random, $2$-local, \textit{all-to-all} Hamiltonian
\[\label{eq:driving_Hamiltonian_random}
\hamd=C \sum_{i < j}^L \sum_{\alpha = x,\, y,\, z} J_{ij}^{\alpha} \hat{\sigma}_i^\alpha \hat{\sigma}_j^\alpha,
\]
where the coupling constants $J_{ij}^{\alpha}$ are randomly extracted from a normal distribution and the normalization factor $C$ is chosen such that $\norm{\hamd - v_{\mathrm{min}}} = 2$, to ensure a fair comparison between different realizations (instances). ($C\propto L^{-3/2}$, which for the SY Hamiltonian follows from the replica formalism~\cite{Bray_1980,PhysRevLett.70.3339}. We numerically confirm this scaling in the Supplementary material.) The results are shown in Fig.~\ref{fig:figure1} (a). We clearly see that the power is bounded by the degree of $k$-locality and \textit{not} by the participation number. As a result, we do not find any extensive charging advantage for this model as expected.
Interestingly, we observe that both the maximum power as well as the maximum value of the commutator norm $$\norm{[\hamz, \, \hamd]}$$ slightly decreases with the system size, $L$.
This is a finite size effect which sensitively reduces by further increasing the system size.
We present an analysis of this phenomenon in the Supplemental Material.

% This result was conjectured, but not proven, by the authors of \cite{}. They were able to prove a looser bound
% \[\label{eq:old_bound}
% \Gamma_{C_0} \leq \gamma k ((m - 1)k + 1),
% \]
% where $m$ denotes the maximum \textit{participation number} of a qubit in the system. 
% The participation number is the number of distinct terms in the summation in eq.~\eqref{eq:interaction_Hamiltonian} that include a given single site $l$.

% It is important to stress that eq.~\eqref{eq:old_bound} by itself does \textit{not} forbid an extensive quantum charging advantage without performing global operations.
% To show that, one can consider a driving Hamiltonian that has a finite and fixed value of $k$ but having \textit{all-to-all} couplings. 
% In such a case, the participation number of a given qubit is of order $m =
% %L^{k - 1}$
% {{L-1}\choose{k-1}}\approx (L-1)^{k-1}/(k-1)!$. Thus, by making use of eq.~\eqref{eq:old_bound} only, one could not exclude the presence of an extensive quantum charging advantage.
% Such a possibility is instead explicitly ruled out by Theorem~\ref{theorem:1} and its corollary.

%\ds{Maybe say something about the classical scaling with $L$ versus the quantum scaling with $L^2$. Maybe this belongs to the introduction.}

%\section{Does global charging always lead to an extensive quantum advantage?}
\emph{Does global charging always lead to an extensive quantum advantage?---}
The presence of a global charging term in $\hamd$ does not guarantee an extensive charging advantage.

As an example, consider a battery composed of $L$ qubits having initial Hamiltonian $\hamz=\sum_{l=1}^{L}h\hat \sigma^z_l$ and charged via the following driving
% \[\label{eq:driving_Hamiltonian_2}
% \hamd=\frac{2V}{L+\frac{3+(-1)^L}{2}}\big(\sum_{l=\textrm{odd}}\hat \sigma_l^x\otimes\hat \sigma_{l+1}^x+\bigotimes^L_{l=1}\hat \sigma_l^x\big),
% \]
\[\label{eq:driving_Hamiltonian_2}
\hamd=\frac{V}{\floor{L/2}+1}\big(\sum_{l=\textrm{odd}}\hat \sigma_l^x\otimes\hat \sigma_{l+1}^x+\bigotimes^L_{l=1}\hat \sigma_l^x\big),
\]
% \[\label{eq:driving_Hamiltonian_2}
% \hamd=\frac{V}{\floor{L/2}+1-\frac{1}{4}\big(1-(-1)^{\floor{L/2}+\floor{(L-1)/2}}\big)}\big(\sum_{l=\textrm{odd}}\hat \sigma_l^x\otimes\hat \sigma_{l+1}^x+\bigotimes^L_{l=1}\hat \sigma_l^x\big),
% \]
with $V$ being a constant. %It is normalized so that $\normt{\hamd-v_{\min}}=V$ for all $L$, $v_{\min}=-\floor{L/2}-1+\frac{1}{2}\big(1-(-1)^{\floor{L/2}+\floor{(L-1)/2}}\big)$, $v_{\max}=\floor{L/2}+1$. %v_\min=-\floor{L/2}-1, initial state either 01010101 or 10101010 for this minimum, depending on whether L is even or odd
From Theorem~\ref{theorem:1}, we obtain
$P\leq 2Lh\norm{\hamd}=2L h V$ (using $\normt{\hat{H}_\mathrm{s}-E_{\mathrm{s}\min}}=2h$), 
%, $P\leq Lh\norm{\hamd-v_{\min}}/2=L h V$
due to the second term representing a global operation, which couples all of the sites at the same time. Presence of this global charging term suggests  possibility of an extensive charging advantage.

However, in this case an extensive quantum advantage is not reached. This is because the nearest-neighbor terms $\sum_{l=\textrm{odd}}\hat \sigma_l^x\otimes \hat \sigma_{l+1}^x$, which provide \textit{non-extensive} advantage, dominate the interaction. They contribute~$V/(1+1/\floor{L/2})$ while the global term contributes only~$V/(\floor{L/2}+1)$ to the total norm $\norm{\hamd}$. This sub-extensive scaling is confirmed by the following, alternative, inequality, which is derived in the Supplementary material,
%On the other hand, the nearest-neighbors terms $\sum_{l=\textrm{odd}}\hat \sigma_l^x\otimes \hat \sigma_{l+1}^x$, each of them providing a \textit{non-extensive} advantage, tend to suggest that an extensive quantum advantage would be likely not reached in this case.
%This intuition is  confirmed by the following, alternative, inequality, which is derived in the Supplementary material,
\[\label{eq:new_inequality}
\abs{P(t)}\leq \sum_{k=1}^{L} k\norm{\hamd_k}\norm{\hat{H}_\mathrm{s}-E_{\mathrm{s}\min}},
\]
where $\hamd_k$ is the $k$-local part of $\hamd$. From inequality~\eqref{eq:new_inequality} we obtain 
$P\leq 4\normt{\hamd_2}h+2L\normt{\hamd_L}h=(\frac{4}{1+1/\floor{L/2}}+\frac{2L}{\floor{L/2}+1})Vh\approx 8Vh$ for the present example, which indeed confirms that the power does not display an extensive advantage. As a further confirmation, we explicitly computed the maximum charging power for this driving Hamiltonian \eqref{eq:driving_Hamiltonian_2}, reported in Fig.~\ref{fig:figure1} (b). We clearly see that the charging power stays well-below the threshold given by eq.~\eqref{eq:new_inequality}.

\emph{Discussion and conclusions.---}
We found a bound on the maximum charging power which can be achieved by charging a quantum battery via an external quench protocol. 

This bound shows that the maximum charging power is not dependent on the operator norm of the battery Hamiltonian, by which the amount of charged energy is measured. Instead, it is governed by the maximum energy difference, $\Delta E$, between two eigenstates of the battery Hamiltonian for which the driving Hamiltonian has a non-vanishing matrix element. In other words, the charging power is limited by the amount of energy that the driving Hamiltonian can add into the battery in a single step, a result which \textit{a posteriori} seems very natural. This bound can be applied to a general quantum battery, described by any Hamiltonian, even those of interacting quantum cells.

When applied to quantum batteries made of $L$ identical cells, this bound provides a limit on how fast they can be charged as compared to classical batteries. The maximum speed by which a quantum battery can be charged depends only on the number $k$ of cells interacting together in a single term. It \emph{does not} depend on the participation number, which is the number of independent terms in the driving Hamiltonian in which a single cell appears. For example, pair-wise interactions can provide a quantum speedup by at most a factor of two,  even in the case of all-to-all couplings, where every cell is connected to every other cell. For a speedup of a factor of $k$, one needs to consider k-cell interactions, while the maximal speedup of $L$ is achieved for $L$-particle interactions. While charging power of classical batteries scales linearly with the number of cells ($\propto L$), quantum batteries provide at most quadratic scaling in charging power ($\propto L^2$). This quadratic scaling cannot be reached without global operations. However, the mere presence of global charging operations does not always guarantee an extensive charging advantage, as we demonstrated on an explicit example.

% Finally, we have shown an alternative bound, eq.~\eqref{eq:new_inequality}, which can bound the maximum charging power more tightly. As a consequence of this bound we were able to show that the mere presence of global charging operations does not always guarantee an extensive charging advantage.

%In classical batteries, the charging power scales linearly with the number of cell, while our result shows that the charging of quantum batteries scales at most quadratically.
This work adds to other  results, in which quantum systems provide at most quadratic improvement over the known classical method, like the Heisenberg limit in sensitivity scaling in quantum metrology over the classically achievable shot-noise limit~\cite{giovannetti2004quantum,giovannetti2011advances,demkowicz2012elusive}, and Grover's search algorithm~\cite{grover1996fast}, which is known to be asymptotically optimal~\cite{bennett1997strengths}.

The bound specifies, for a given battery Hamiltonian and for a given driving, the maximum instantaneous charging power achievable in that particular setup.
It does not give any information about the quantum state for which such a power can be achieved.
This constitutes an interesting question for future research.

\emph{Acknowledgements.—--}
We acknowledge the support by the Institute
for Basic Science in Korea (IBS-R024-Y2
 and IBS-R024-D1). 
DR would like to thank M.~Carrega, J.~Kim, J.~Murugan and J. Olle for collaboration on related projects.

%\bibliography{bibliography.bib}
    
\clearpage

\setcounter{section}{0}
\setcounter{theorem}{0}
\setcounter{corollary}{0}

\section*{Supplemental Material for: \\
``Quantum Charging Advantage Cannot Be Extensive Without Global Operations''
}

\section{\label{sec:level1}PROOF OF theorem 1}
Let us consider an initial Hamiltonian, $\hamz$, which is assumed to be time independent, having spectral decomposition $\hamz=\sum_{j=1}^{N} E_{j}\ket{E_{j}}\bra{E_j}$, where the energies $E_j$ are ordered, \textit{i.e.} $E_j \leq E_m$ for $j < m$.
Let us consider a different Hamiltonian $\hamd$, which we call the driving Hamiltonian, which we will write in basis of the initial Hamiltonian as $\hamd=\sum_{j , m=1}^{N} V_{jm}\ket{E_{j}}\langle E_m\rvert$. $\Delta E$ is the minimum number such that for all $j$ and $m$,
% \[\label{eq:Vdefinition}
% V_{jm} = 0 \quad \mathrm{when} \quad  | E_{j}-E_{m} | >\Delta E, 
% \]
\[\label{eq:Vdefinition}
\mathrm{when}\quad | E_{j}-E_{m} | >\Delta E,\quad \mathrm{then}\quad V_{jm} = 0.
\]
We also consider the spectral decomposition of the driving Hamiltonian as $\hamd=\sum_{\alpha}\alpha\pro{\alpha}{\alpha}$. 
Out of $\hamz$ and $\hamd$ we define other two operators $\hat{h}(e)$ and $\hat{v}(e)$, functions of a continuous variable $e$ satisfying $0 < e \leq \Delta E$. Explicitly, the definition of $\hat{h}(e)$ takes the following form
\begin{eqnarray}
\hamh{e}  &=& \sum_{j=1}^{N} h_{j} (e)\ket{E_{j}}\bra{E_{j}} \ ,
\end{eqnarray}
where the functions $h_{j} (e)$ are defined as follows
\begin{eqnarray}\label{eq:definitionhj}
h_{j}(e) =\begin{cases}
\frac{1}{2}\quad &\mathrm{for}\ \floor{\frac{E_{j}-e}{\Delta E}}\ \mathrm{is\ odd}\\
-\frac{1}{2}\quad &\mathrm{for}\ \floor{\frac{E_{j}-e}{\Delta E}}\ \mathrm{is\ even}\\
\end{cases}.
\end{eqnarray}

By definition, $\hat{h}(e)$ has the same eigenvectors as $\hamz$ but its eigenvalues are restricted to be either $-\frac{1}{2}$ or $\frac{1}{2}$. The eigenvalues of $\hat{h}(e)$, as functions of $e$, are depicted in Fig.~\ref{Fig:distribution_zeros}.
In turn, $\hamv{e}$ is defined using $\hamh{e}$ as follows
\begin{eqnarray}\label{eq:vdefinition}
\hamv{e} &=& \sum_{j,m=1}^{N} v_{jm}(e)\ket{E_{j}}\bra{E_{m}}.\\
v_{jm}(e)&=&\begin{cases}
V_{jm}\ \quad\,\,\mathrm{for}\ (j-m)(h_{j}(e)-h_{m}(e))\geq 0\\
-V_{jm}\quad\mathrm{for}\ (j-m)(h_{j}(e)-h_{m}(e))< 0
\end{cases}\notag
\end{eqnarray}
By their definition, we can deduce the following lemmas for $\hamv{e}$ and $\hamh{e}$.
% \begin{figure}[t]
% \begin{center}
% %\includegraphics[width=1\hsize]{evolution7.pdf}\\
% \begin{tikzpicture}{\label{fig:spectrumhl}}
% \node[left] at (-0.3,0) {$0$};
% \node[left] at (-0.3,2) {$\Delta E$};
% \node[left] at (-0.3,1) {$e$};
% \node at (4.5,1) {$\cdots$};
% \draw (0,0) -- (4,0);
% \draw (0,2) -- (4,2);
% \draw (5,0) -- (7,0);
% \draw (5,2) -- (7,2);
% \draw (0,0) -- (2,2);
% \draw (2,0) -- (4,2);
% \draw (7,1.7) -- (5.3,0);
% \draw (5.3,2) -- (5,1.7);
% \node[below] at (0,0) {$\ket{E_1}$};
% \node[below] at (0.9,0) {$\ket{E_{2}}$};
% \node[below] at (0.9,-0.5) {$\vdots$};
% \node[below] at (0.9,-0.8) {$\ket{E_{j}}$};
% \node[below] at (2.2,0) {$\ket{E_{j+1}}$};
% \node[below] at (2.2,-0.5) {$\vdots$};
% \node[below] at (2.2,-0.8) {$\ket{E_{j+m}}$};
% \node[below] at (7,0) {$\ket{E_{N}}$};
% \node[below] at (5.9,0) {$\ket{E_{N-k}}$};
% \node[below] at (5.9,-0.5) {$\vdots$};
% \node[below] at (5.9,-0.8) {$\ket{E_{N-1}}$};
% \draw[dashed] (0.9,0) -- (0.9,2);
% \draw[dashed] (2.2,0) -- (2.2,2);
% \draw[dashed] (0,0) -- (0,2);
% \draw[dashed] (7,0) -- (7,2);
% \draw[dashed] (5.9,0) -- (5.9,2);
% \node at (0.5,1.3) {{\large $\frac{1}{2}$}};
% \node at (2,1) {{\large $-\frac{1}{2}$}};
% \node at (3.5,0.7) {{\large $\frac{1}{2}$}};
% \node at (5.3,1) {{\large $-\frac{1}{2}$}};
% \node at (6.8,0.7) {{\large $\frac{1}{2}$}};
% \end{tikzpicture}
% \caption
% {distributions of zeros and ones over states. x axis means Energy of basis in $\hamz$ and y axis means variable $e$ of $\hamh{e}$.}
% \label{Fig:distribution_zeros}
% \end{center}
% \end{figure}

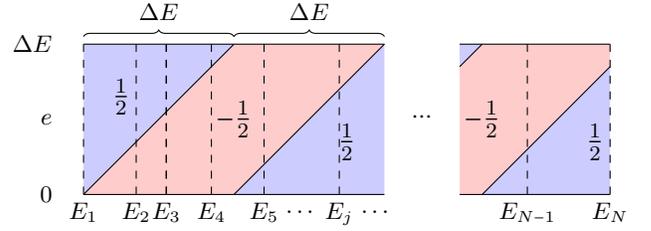
\begin{figure}[t]
\begin{center}
\begin{tikzpicture}{\label{fig:spectrumhl}}
\fill[fill=blue, fill opacity=0.2] (0, 0) -- (0, 2) --  (2, 2)  -- cycle;
\fill[fill=red, fill opacity=0.2] (0, 0) -- (2, 0) -- (4,2) -- (2, 2)  -- cycle;
\fill[fill=blue, fill opacity=0.2] (2, 0) -- (4, 0) --  (4, 2)  -- cycle;
\fill[fill=blue, fill opacity=0.2] (5, 2) -- (5.3, 2) -- (5,1.7) -- cycle;
\fill[fill=red, fill opacity=0.2] (5, 0) -- (5.3, 0) -- (7,1.7) -- (7, 2) -- (5.3,2) -- (5,1.7) -- cycle;
\fill[fill=blue, fill opacity=0.2]  (5.3, 0) -- (7,0) -- (7, 1.7) -- cycle;
\node[left] at (-0.3,0) {$0$};
\node[left] at (-0.3,2) {$\Delta E$};
\node[left] at (-0.3,1) {$e$};
\node at (4.5,1) {$\cdots$};
\draw[decorate, decoration = {brace}] (0,2.1) -- (2,2.1);
\draw[decorate, decoration = {brace}] (2,2.1) -- (4,2.1);
\node[above] at (1,2.2) {$\Delta E$};
\node[above] at (3,2.2) {$\Delta E$};
\draw (0,0) -- (4,0);
\draw (0,2) -- (4,2);
\draw (5,0) -- (7,0);
\draw (5,2) -- (7,2);
\draw (0,0) -- (2,2);
\draw (2,0) -- (4,2);
\draw (7,1.7) -- (5.3,0);
\draw (5.3,2) -- (5,1.7);
\node[below] at (0,0) {$E_1$};
\node[below] at (0.7,0) {$E_{2}$};
\node[below] at (1.1,0) {$E_{3}$};
\node[below] at (1.7,0) {$E_{4}$};
\node[below] at (2.4,0) {$E_{5}$};
\node[below] at (2.9,-0.1) {$\dots$};
\node[below] at (3.4,0) {$E_{j}$};
\node[below] at (3.9,-0.1) {$\dots$};
\node[below] at (7,0) {$E_{N}$};
\node[below] at (5.9,0) {$E_{N-1}$};
\draw[dashed] (0.7,0) -- (0.7,2);
\draw[dashed] (1.1,0) -- (1.1,2);
\draw[dashed] (1.7,0) -- (1.7,2);
\draw[dashed] (2.4,0) -- (2.4,2);
\draw[dashed] (3.4,0) -- (3.4,2);
\draw[dashed] (1.1,0) -- (1.1,2);
\draw[dashed] (0,0) -- (0,2);
\draw[dashed] (7,0) -- (7,2);
\draw[dashed] (5.9,0) -- (5.9,2);
\node at (0.5,1.3) {{\large $\frac{1}{2}$}};
\node at (2,1) {{\large $-\frac{1}{2}$}};
\node at (3.5,0.7) {{\large $\frac{1}{2}$}};
\node at (5.3,1) {{\large $-\frac{1}{2}$}};
\node at (6.8,0.7) {{\large $\frac{1}{2}$}};
\end{tikzpicture}
\caption
{Values $\hat{h}_j(e)$ ($1/2$ as the blue region, and $-1/2$ as the red region) as a function of index $j$, which relates to eigenstate energy $E_j$ depicted on the $x$-axis, and of parameter $e$ depicted on the $y$-axis. Condition~\eqref{eq:Vdefinition} means that for energies $E_j$ and $E_m$ that are further than $\Delta E$, the corresponding element of the driving Hamiltonian $V_{jm}$ is zero.}
\label{Fig:distribution_zeros}
\end{center}
\end{figure}

%\subsection{theorem 1}
\begin{lemma}\label{lemma:1}
We can express the commutator between the initial and the driving Hamiltonians using $\hamv{e}$ and $\hamh{e}$ as follows:
\begin{eqnarray}
\left[\hamz,\hamd\right]&=&\int_{e=0}^{\Delta E}{[\hamh{e},\hamv{e}]de}.\label{eq:statement3}
\end{eqnarray}
\end{lemma}

\begin{proof}\label{pro:lem_1}
Let us write explicitly the commutator of initial Hamiltonian ,$\hamz$, with driving Hamiltonian ,$\hamd$, as  
\begin{eqnarray}\label{eq:commuteHV}
\left[\hamz,\hamd\right]&=&\sum_{j,m=0}^{N}{[\hamz,V_{jm} \ket{E_{j}}\bra{E_{m}}]}\notag\\
&=&\sum_{j,m=0}^{N}(E_{j}-E_{m})V_{jm} \ket{E_{j}}\bra{E_{m}}
\end{eqnarray}
as well as the commutator of operator $\hamh{e}$ with operator $\hamv{e}$ 
\[\label{eq:commutehv}
\begin{split}
&\int_{e=0}^{\Delta E} [\hamh{e},\hamv{e}] de\\
&=\int_{e=0}^{\Delta E}\sum_{j,m=0}^{N} [\hamh{e},v_{jm}(e)\ket{E_{j}} \bra{E_{m}}]de\\
&=\sum_{j,m=0}^{N}\bigg(\int_{e=0}^{\Delta E}{(h_{j}(e)-h_{m}(e))v_{jm}(e)de}\bigg)\ket{E_{j}}\bra{E_{m}}.
\end{split}
\]
We prove this lemma by comparing elements \[a_{jm}=(E_{j}-E_{m})V_{jm}\] in eq.~\eqref{eq:commuteHV} and \[b_{jm}=\int_{e=0}^{\Delta E}{(h_{j}(e)-h_{m}(e))v_{jm}(e)de}\] eq.~\eqref{eq:commutehv}, and showing that they are equal.

When a pair of $j,m$ satisfies $|E_{j}-E_{m}|>\Delta E$, $V_{jm}$ and $v_{jm}(e)$ are zero by definition(eqs.~\eqref{eq:Vdefinition} and~\eqref{eq:vdefinition}), which means that the corresponding elements $a_{jm}$ and $b_{jm}$ are both zero and thus equal. 

In the opposite case, \textit{i.e.} when  $|E_{j}-E_{m}| \leq \Delta E$, $V_{jm}$ and $v_{jm}(e)$ may not be zero. %They are $(E_{j}-E_{m})v_{jm}(e)$ and $\int_{e=0}^{\Delta E}{(h_{j}(e)-h_{m}(e))v_{jm}(e)de}$. 
%If $j$ is equal to $m$, both elements $a_{jm}$ and $b_{jm}$ are zero and thus equal. 
%Therefore, we focus on the case when $j\neq m$ and $E_{j}-E_{m}\leq\Delta E$. 
In this case, we explicitly have
\[
\begin{split}
&(h_{j}(e)-h_{m}(e))v_{jm}(e)\label{eq:hlhvinde}\\
&=\begin{cases}
(h_{j}(e)-h_{m}(e))V_{jm} &\mathrm{for}\ j> m,\,h_{j}(e)>h_{m}(e)\\
(h_{m}(e)-h_{j}(e))V_{jm} &\mathrm{for}\ j>m,\,h_{j}(e)<h_{m}(e)\\
(h_{m}(e)-h_{j}(e))V_{jm} &\mathrm{for}\ j< m,\,h_{j}(e)>h_{m}(e)\\
(h_{j}(e)-h_{m}(e))V_{jm} &\mathrm{for}\ j< m,\,h_{j}(e)<h_{m}(e)\\
0 &\mathrm{for}\ h_{j}(e)=h_{m}(e)
\end{cases}\\
&=\begin{cases}
\ V_{jm}&\mathrm{for}\ j>m,\ h_{j}(e)\neq h_{m}(e)\\
-V_{jm}&\mathrm{for}\ j<m, \ h_{j}(e)\neq h_{m}(e)\\
0&\mathrm{for}\ h_{j}(e)= h_{m}(e)
\end{cases}
\end{split}
\]
%From eq.~\eqref{eq:hlhvinde} we deduce that when $(h_{j}(e)-h_{m}(e))v_{jm}(e)$ is not zero, it just depends on $j-m$ but not on $h_{j}(e)-h_{m}(e)$.

%Since $(h_{j}(e)-h_{m}(e))v_{jm}(e)$, that is equal to $\pm V_{jm}$, is independent on $e$, to calculate the integral of $(h_{j}(e)-h_{m}(e))v_{jm}(e)$ 
For a fixed $j,m$, this equation implies that the integrand $(h_{j}(e)-h_{m}(e))v_{jm}(e)$ is a piecewise constant function, which is equal to either $0$ and $V_{jm}$ when $j > m$, or $0$ and $-V_{jm}$ when $j < m$. To figure out at which points this function jumps, we need to take a look at the definition of $h_{m}(e)$, eq.~\eqref{eq:definitionhj}. Then to calculate $b_{jm}$, we just need to integrate over this piecewise function, which is straightforward once we know where it jumps.

In order to identify this jump, we need to determine in which region, in the interval $0\leq e<\Delta E$, $h_{j}(e)$ and $h_{m}(e)$ are not equal. 
First, we assume that $j > m$. Hence, $E_j$ is also bigger than $E_m$, since we are assuming that $E_j$ are ordered. Let us rewrite $E_m$  as 
\begin{eqnarray}
\label{eq:integer_decomposition}
E_m = n_m\Delta E+x_m
\end{eqnarray}
for the maximal integer $n_m$ such that $x_m$ is a real number in the interval $0\leq x<\Delta E$. 
This shows
\[
\floor{\frac{E_m-e}{\Delta E}}=\floor{\frac{n_m\Delta E+x_m-e}{\Delta E}}=
\begin{cases}
n_m &\mathrm{for\ } x_m\geq e\\
n_m-1 &\mathrm{for\ } x_m< e
\end{cases}
\]
which implies that $h_m$ is a piecewise constant function made from two pieces with the jump at $e=x$.

Given that $E_j$ and $E_m$ differ at most by $\Delta E$, we conclude that in the decomposition 
\[
E_j = n_j \Delta E + x_j \ ,
\]
$n_j$ can be either equal to $n_m$ or $n_m + 1$.
Let us first consider the case $n_j = n_ m$. In this case (case 1), we have
\[
\floor{\frac{E_j-e}{\Delta E}}=\floor{\frac{n_m\Delta E+x_j-e}{\Delta E}}=
\begin{cases}
n_m &\mathrm{for\ } x_j\geq e\\
n_m-1 &\mathrm{for\ } x_j< e
\end{cases}
\]
Noticing that our assumption $E_j>E_m$ implies that $x_j>x_m$, we have
\[
h_{j}(e)-h_{m}(e)=
\begin{cases}
0 &\mathrm{for\ } e \leq x_m\\
(-1)^{n_m+1} &\mathrm{for\ } x_m < e \leq x_j \\
0 &\mathrm{for\ } e > x_j
\end{cases}
\]
In the opposite case (case 2), \textit{i.e.} when $n_j = n_m + 1$, our assumption $E_j>E_m$ implies $x_j<x_m$, and we have
\[
h_{j}(e)-h_{m}(e)=
\begin{cases}
(-1)^{n_m} &\mathrm{for\ } e \leq x_j\\
0 &\mathrm{for\ } x_j < e \leq x_m \\
(-1)^{n_m+1} &\mathrm{for\ } e > x_m
\end{cases}
\]
Then to get the full integrand, we plug in the $v_{jm}(e)$ again, whose sole role is to transform the negative sign into the positive sign and to add $V_{jm}$. This in case 1 gives
\[
(h_{j}(e)-h_{m}(e))v_{jm}(e)=
\begin{cases}
0 &\mathrm{for\ } e \leq x_m\\
V_{jm} &\mathrm{for\ } x_m < e < x_j \\
0 &\mathrm{for\ } e \geq x_j
\end{cases}
\]
while in case 2 we obtain
\[
(h_{j}(e)-h_{m}(e))v_{jm}(e)=
\begin{cases}
V_{jm} &\mathrm{for\ } e \leq x_j\\
0 &\mathrm{for\ } x_j < e < x_m \\
V_{jm} &\mathrm{for\ } e \geq x_m
\end{cases}
\]
% This gives
% \[
% x=E_m-n\Delta E
% \]
% and also
% \[
% h_m(e) =\begin{cases}
% \frac{1}{2}\quad &\mathrm{for}\ \floor{\frac{E_{j}-e}{\Delta E}}\ \mathrm{is\ odd}\\
% -\frac{1}{2}\quad &\mathrm{for}\ \floor{\frac{E_{j}-e}{\Delta E}}\ \mathrm{is\ even}\\
% \end{cases}.
% \]
In case 1, we compute
\[
b_{jm}=V_{jm}(x_j-x_m)=V_{jm}(E_j-E_m)=a_{jm}.
\]
In case 2, we compute
\[
b_{jm}=V_{jm}(x_j+\Delta E-x_m)=V_{jm}(E_j-E_m)=a_{jm}.
\]
In both cases, we have shown that $a_{jm}=b_{jm}$. The case of $m>j$ follows from identical arguments. This concludes the proof.

\end{proof}

\begin{lemma}\label{lemma:2}
Operator norms of $\hamv{e}$ and $\hamd$ are equal for every $e$, expressed mathematically as
\begin{eqnarray}
\norm{\hamv{e}}&=&\norm{\hamd}\label{eq:statement2}.
\end{eqnarray}
\end{lemma}

\begin{proof}\label{pro:lem_2}
We prove lemma \ref{lemma:2} by showing that $\hamv{e}$ and $\hamd$ have same spectra. Then the Lemma follows directly from the definition of the operator norm.
Operator norm is defined as
\[
\norm{\hat A}=\sup_{\ket{\psi}}\frac{\norm{\hat A\ket{\psi}}}{\norm{\ket{\psi}}}
\]
where the supremum goes over all vectors $\ket{\psi}$. In the case of bounded Hermitian operators, this is equivalent to the supremum of absolutes value of eigenvalues of the operator $\hat A$,
\[\label{eq:corresponding_definition_of_op_norm}
\norm{\hat A}=\sup_{\L}\abs{\L},
\]
which means that the same spectrum directly imply the same operator norm between different operators. 

We express each eigenvector $\ket{\A}$ of the driving Hamiltonian $\hamd$ in terms of eigenvectors of the initial Hamiltonian $\hamz$ as 
\begin{eqnarray}
\ket{\A}=\sum_{j=1}^{N}{\B_{j}\lvert E_{j}\rangle}.
\end{eqnarray}
By definition of an eigenstate, the following equation holds:
\begin{eqnarray}
\A\sum_{j=1}^{N}{\B_{j}\lvert E_{j}\rangle}=\A\ket{\A}=\hamd\ket{\A}=\sum_{j,m=0}^{N}{V_{jm}\B_{m}\lvert E_{j}\rangle}.
\end{eqnarray}
%Eq~\eqref{eq:property_eigen} produces a statement that, if and only if $\A\A_j$ is equal to $\sum_{m=0}^{N}{V_{jm}\A_{m}}$ for all j, $\ket{\A}$ is eigenstate of $\hamd$ and eigenvalue is $\A$.
Therefore, the fact that $\ket{\alpha}$ is an eigenstate of driving Hamiltonian $\hamd$ is equivalent to that for all $j$,
\[\label{eq:property_eigen}
\A\B_j=\sum_{m=1}^{N}{V_{jm}\B_{m}}.
\]

Further, for each eigenvalue of the driving Hamiltonian $\alpha$, and for any arbitrary but fixed value of $e$, we define the state $\ket{\G(e)}$ (will we omit writing the dependence on $\alpha$ for clarity), which we will prove to be an eigenvector of operator $\hamv{e}$ that corresponds to the same eigenvalue (but now of operator $\hamv{e}$) $\alpha$. Thus, by doing this we will show that $\hamd$ and $\hamv{e}$ have the same spectrum. We define state $\ket{\G(e)}$ as 
\[
\ket{\G(e)}=\sum_{j=1}^{N}\G_j(e)\ket{E_j}.
\]
Coefficients $\G_j(e)$ are defined by a recurrence relation
\[\label{eq:definition_G}
\begin{split}
&\G_{1}(e) =\B_{1}\\
&\G_{j+1}(e) =\begin{cases}
\ \G_{j}(e)\frac{\B_{j+1}}{\B_{j}},\quad &\mathrm{for}\ h_{j+1}(e)\geq h_{j}(e),\\
-\G_{j}(e)\frac{\B_{j+1}}{\B_{j}},\quad &\mathrm{for}\ h_{j+1}(e)< h_{j}(e).
\end{cases}
\end{split}
\]
Since $V_{jm}$ and $v_{jm}(e)$ are zero when $|E_j-E_m|>\Delta E$ (as follows from eqs.~\eqref{eq:Vdefinition} and~\eqref{eq:vdefinition}), we just need to consider pairs of $j,m$ that satisfy $|E_j-E_m|\leq\Delta E$. 

First, we assume $j>m$. We will show that there is at most one index $k$, $j> k\geq m$, for which $h_{k+1}(e)$ and $h_{k}(e)$ are different. This statement is relatively clear from Fig.~\ref{Fig:distribution_zeros}: for $|E_j-E_m|\leq\Delta E$ and a fixed $e$, the horizontal line drawn at height of $e$ can cross the red region into blue (or blue into red) at most once. To prove this statement mathematically (which duplicates the visual proof), for a contradiction we assume that there are two distinct indexes $k$ and $k'$ such that $j>k'>k\geq m$, where the sign changes, i.e., for which $h_{k+1}(e)\neq h_{k}(e)$ and $h_{k'+1}(e)\neq h_{k'}(e)$. By definition of $h_k(e)$ (eq.~\eqref{eq:definitionhj}),
\[\label{eq:biggerthan2delta}
E_{k'+1}-E_{k}> \Delta E.
\]
However, from $j>k'>k\geq m$ and from the fact that we assumed our energy levels to be ordered, we have
\[
E_j-E_m\geq E_{k'+1}-E_{k}
\]
which according to eq.~\eqref{eq:biggerthan2delta} means that $E_j-E_m>\Delta E$, which is in contradiction with our assumption $|E_j-E_m|\leq\Delta E$.

% \[
% \begin{split}
% E_j-E_m&\geq E_{k'+1}-E_{k}\\
% &=(\Delta E(n_{k'+1})+x_{k'+1}+e)-(\Delta En_k+x_k+e))\\
% &=(\Delta E(n_{k'}+1)+x_{k'+1}+e)-(\Delta En_k+x_k+e))\\
% &=\Delta E+\Delta E(n_{k'}-n_k)+x_{k'+1}-x_k\\
% &>\Delta E,
% \end{split}
% \]
% where we have used that $n_{k'+1}=n_{k'}+1$ and $\abs{x_{k'+1}-x_k}<\Delta E$, and $n_{k'}\geq n_{k}+1$,
% which is in contradiction with the assumption that  $|E_j-E_m|\leq\Delta E$.

So we just showed that for every combination $j>m$ there exists at most one index $k$, $j>k\geq m$, such that $h(e)$ changes sign at this index, i.e., $h_k(e) \neq h_{k+1}(e)$. Using this knowledge, we compute a ratio that will be useful later in showing that  $\ket{\gamma(e)}$ is an eigenstate of $\hamv{e}$.
In the case when $k$ exists, we obtain ratio
\[
\begin{split}
&\frac{\G_{m}(e)}{\G_{j}(e)}=\prod_{s=m}^{j-1} \frac{\G_{s}(e)}{\G_{s+1}(e)}\\
&=\begin{cases}
\;\;\;\prod_{s=m}^{j-1} \frac{\B{s}}{\B{s+1}}=\;\frac{\B{m}}{\B{j}}&\quad\mathrm{for}\; h_{k+1}(e)-h_{k}(e)>0\\
-\prod_{s=m}^{j-1} \frac{\B{s}}{\B{s+1}}=-\frac{\B{m}}{\B{j}} &\quad\mathrm{for}\; h_{k+1}(e)-h_{k}(e)<0
\end{cases}
\end{split}
\]
by the definition of $\G$ (eq.~\eqref{eq:definition_G}).

Because $k$ is the only point where value $h(e)$ changes, it means that $h_j(e)=h_{k+1}(e)$ and $h_m(e)=h_{k}(e)$, which means we can rewrite the above equation as
\[
\begin{split}
&\frac{\G_{m}(e)}{\G_{j}(e)}
=\begin{cases}
\ \frac{\B{m}}{\B{j}} &\quad\mathrm{for}\; h_{j}(e)-h_{m}(e)> 0\\
-\frac{\B{m}}{\B{j}} &\quad\mathrm{for}\; h_{j}(e)-h_{m}(e)<0.
\end{cases}
\end{split}
\]
If such $k$ does not exist, then $h_m(e)=h_{m+1}(e)=\cdots=h_j(e)$, in which case $\G_{m}(e)/\G_{j}(e)=\B_m/\B_j$, so by including this case we can generalize this the above equation to
\[
\begin{split}
&\frac{\G_{m}(e)}{\G_{j}(e)}
=\begin{cases}
\ \frac{\B{m}}{\B{j}} &\quad\mathrm{for}\; h_{j}(e)-h_{m}(e)\geq 0\\
-\frac{\B{m}}{\B{j}} &\quad\mathrm{for}\; h_{j}(e)-h_{m}(e)<0.
\end{cases}
\end{split}
\]

The case of $m>j$ follows from identical arguments, which result in a change of sign. Combining cases $j>m$, $j<m$ and $j=m$ (which is trivial) together, we finally obtain the ratio for any $j$ and $m$ as
\begin{equation}\label{eq:theratio}
\frac{\G_{m}(e)}{\G_{j}(e)}=
\begin{cases}
\ \frac{\B_{m}}{\B_{j}}\quad\,\,\mathrm{for}\ (h_{j}(e)-h_{m}(e))(j-m)\geq 0\\
-\frac{\B_{m}}{\B_{j}}\quad\mathrm{for}\ (h_{j}(e)-h_{m}(e))(j-m)<0.
\end{cases}
\end{equation}
%By definition of $v_{jm}(e)$, eq.~\eqref{eq:vdefinition}, and by eq.~\eqref{eq:aldival},
%By plugging in eqs.~\eqref{eq:vdefinition} and~\eqref{eq:aldival} we obtain
Combining eq.~\eqref{eq:vdefinition}, \eqref{eq:property_eigen} and eq.~\eqref{eq:theratio} we obtain
\[
\sum_{m=1}^{N}{v_{jm}(e)\frac{\G_{m}(e)}{\G_{j}(e)}}=\sum_{m=1}^{N}{V_{jm}\frac{\B_{m}}{\B_{j}}}=\alpha.
\]
Using this property, we have
\begin{equation}
\begin{split}
\hamv{e}\ket{\G(e)}&=\sum_{j,m=1}^{N}v_{jm}(e)\G_{m}(e)\ket{E_{j}}\\
&=\A\sum_{j=1}^{N}{\G_{j}(e)\lvert E_{j}\rangle}=\A\ket{\G(e)}.
\end{split}
\end{equation}
This directly implies that $\ket{\G(e)}$ is an eigenstate of $\hamv{e}$ associated with eigenvalue $\A$. Therefore, $\hamd$ and $\hamv{e}$ have the same spectrum. By Eq.~\eqref{eq:corresponding_definition_of_op_norm}, $\norm{\hamd}=\norm{\hamv{e}}$, which concludes the proof.
%Any arbitrary eigenstate of $\hamd$ has the pair that is eigenstate of $\hamv{e}$ conserving eigenvalue. That means $\hamd$ and $\hamv{e}$ have same spectrum.
\end{proof}

\begin{theorem}
\label{theo:main_theo}
(Extended version of Theorem~1.~shown in the main text) For driving that couples energy levels with at most $\Delta E$ energy difference, as expressed by eq.~\eqref{eq:Vdefinition}, the following series of inequalities holds.
\[
\norm{[\hamz,\hamd]}\leq\Delta E\inf_{\hat{D}}\norm{\hamd\!-\!\hat{D}}\leq\Delta E\norm{\hamd-v_{\min}}/2\leq\Delta E\norm{\hamd},
\]
where $\hat{D}$ are matrices diagonal in the eigenbasis of the initial Hamiltonian $\ham$.
\end{theorem}
\begin{proof}
By definition~\eqref{eq:definitionhj} the eigenvalues of $\hamh{e}$ are either $\frac{1}{2}$ or $-\frac{1}{2}$. Hence 
\[\label{eq:statement1}
\norm{\hamh{e}}=\frac{1}{2}
\]
holds for all $e$'s. 
Using lemma \ref{lemma:1}, $\norm{[\hamz,\hamd]}$ can be rewritten as
\[
\norm{[\hamz,\hamd]}\\=\norm{\int_{e=0}^{\Delta E}{[\hamh{e},\hamv{e}]de}}.
\]
We apply the triangle inequality of operator norm two times (once for the first and once for the second inequality) and obtain
%yields, we obtain, implies, follows, gives, (results in)
\[\begin{split}
&\norm{\int_{e=0}^{\Delta E}{[\hamh{e},\hamv{e}]de}}\\
&\leq\int_{e=0}^{\Delta E}{\norm{[\hamh{e},\hamv{e}]}de}\leq\int_{e=0}^{\Delta E}{2\norm{\hamh{e}}\norm{\hamv{e}}de}.
\end{split}
\]
By lemma \ref{lemma:2} and eq.~\eqref{eq:statement1}, the right hand side of this equation equals
\[\label{eq:boundwithV}
\int_{e=0}^{\Delta E}2\norm{\hamh{e}}\norm{\hamv{e}}de=\int_{e=0}^{\Delta E}{\norm{\hamd}de}=\Delta E\norm{\hamd},
%\end{split}
\]

Since any diagonal matrix $\hat{D}=\sum_{j=1}^{N}\lambda_j\ket{E_j}\bra{E_j}$, commutes with $\hamz$, we can make this bound tighter by minimizing over $\hat{D}$. The key is to realize that if we define $\hamd'=\hamd-\hat{D}$, where $\hamd$ satisfies eq.~\eqref{eq:Vdefinition}, then also $\hamd'=\sum_{j,m=1}^{N}(V_{jm}-\delta_{jm}\lambda_j)\ket{E_j}\bra{E_m}$ satisfies the same equation, with the same $\Delta E$. Thus, equation~\eqref{eq:boundwithV} holds also for operator $\hamd'$, and we have
\[
\norm{\left[\hamz,\hamd\right]}=\norm{\left[\hamz,\hamd-\hat{D}\right]}\leq \Delta E\norm{\hamd-\hat{D}},
\]
which holds for any $\hat{D}$ which is diagonal in the energy basis.

We can take the infimum over all such diagonal matrices, which gives
\[
\norm{\left[\hamz,\hamd\right]}=\inf_{\hat{D}}\norm{\left[\hamz,\hamd-\hat{D}\right]}\leq \Delta E\inf_{\hat{D}}\norm{\hamd-\hat{D}}
\]
This represents the tightest bound obtained by the present method.

To prove the theorem as it is written in the main text, we restrict ourselves to the case of $\hat{D}=\lambda\hat{I}$. In this case we can evaluate the infimum and obtain explicitly
\[
\begin{split}
&\norm{\left[\hamz,\hamd\right]}\leq 
\Delta E\inf_{\hat{D}}\norm{\hamd\!-\!\hat{D}}\\
&\leq\Delta E\inf_{\lambda}\norm{\hamd\!-\!\lambda\hat{I}}=\Delta E\!\norm{\hamd\!-\!v_{\min}}/2,
\end{split}
\]
which proves the theorem.

% The theorem then follows from eq.~\eqref{eq:most_general}.

% For any constant value $\lambda$, commutator $[\hamd+\lambda,\hamz]$ has same value, \textit{e.i.} $\forall \lambda,\quad [\hamd,\hamz]=[\hamd+\lambda,\hamz]$. we obtain inequality of norm of commutator for all $\lambda$
% \begin{equation}\label{eq:preserve_inequality}
% \norm{[\hamd,\hamz]}=\norm{[\hamd+\lambda,\hamz]}\leq\Delta E \norm{\hamd+\lambda}.
% \end{equation}
% By definition of operator norm, norm of driving Hamiltonian is equal to $\norm{\hamd}=\max(\A_{\max},-\A_{\min})$. operator norm of 
% \begin{equation}
% \norm{\hamd+\lambda}=\max(\A_{\max}+\lambda,-\A_{\min}-\lambda)
% \end{equation}
% is minimized as $\frac{\A_{\max}-\A_{\min}}{2}$. Since eq~\eqref{eq:preserve_inequality} is preserved for all $\lambda$, we obtain
% \begin{equation}
% \norm{[\hamd,\hamz]}\leq\Delta E \frac{\A_{\max}-\A_{\min}}{2}=\Delta E\norm{\hamd-V_0}/2.
% \end{equation}
% which proves the theorem.
\end{proof}

\section{Lattice case}

%Let us consider a battery composed of $L$ identical cells, which has the initial Hamiltonian $\hamz=\sum_{l=1}^{L}\haml$, where $\haml$ is initial Hamiltonian of single cell. 
We consider a battery composed of $L$ cells described by initial Hamiltonian 
\[\label{eq:initial_lattice_Hamiltonian}
\hamz=\sum_{l=1}^{L}\haml,
%\hamz=H^{(1)}\otimes\hat{I}\otimes\cdots\otimes\hat{I}\ +\ \cdots\ +\ \hat{I}\otimes\cdots\otimes\hat{I}\otimes H^{(1)}
\]

 We charge this battery by turning on the driving Hamiltonian,
%The most general form for a driving Hamiltonian of a local interaction quantum battery is as follows
\begin{eqnarray}\label{eq:interaction_Hamiltonian}
\hamd(t)=\sum_{\bi\in K(L,k)}{\hat{V}_{\bi}}(t),
\end{eqnarray}
where, by definition, each term in the summation couples together at most $k$ cells. Expressed mathematically,
\begin{align}
K(L,k)&=\bigcup_{n=1}^k C(L,n),\\
C(L,n)&=\{(i_1,\dots,i_n)| i_1<\cdots<i_n \text{\ and\ }i_j\in \{1,\dots,L\}\},\nonumber
\end{align}
% where $C(L,n)$ is a set of all combinations of $n$ sites, and $[\haml,\hamd_\bi(t)]=0$ when $l\not\in\bi=(i_1,\dots,i_n)$. \ds{I think if we would define it in a more standard way, we would say that $V_\bi$ acts as an identity on the site which does not appear in the index, i.e., for any local matrix $\hat{M}^{(l)}=\hat{I}\otimes\cdots\hat{I}\otimes\hat{M}\otimes\hat{I}\cdots\hat{I}$, where $\hat{M}$ is at the $l$-th site, if $l\not\in\bi=(i_1,\dots,i_n)$, then}
% \[
% [\hat{M}^{(l)},\hamd_\bi(t)]=0.
% \]
% This is a more strict definition than the previous one, meaning that if $\hamd$ satisfies this definition, it also satisfies the previous definition.
where $C(L,n)$ is a set of all combinations of $n$ sites, 
and $V_\bi$ acts as an identity on the site which does not appear in the index, i.e., for any local matrix $\hat{M}^{(l)}=\hat{I}\otimes\cdots\otimes\hat{I}\otimes\hat{M}\otimes\hat{I}\otimes\cdots\otimes\hat{I}$, where $\hat{M}$ is at the $l$-th place, if $l\not\in\bi=(i_1,\dots,i_n)$, then 
\[\label{eq:klocality}
[\hat{M}^{(l)},\hamd_\bi(t)]=0.
\]

\begin{corollary}\label{corol:2}
For identical cells, which have initial Hamiltonian $\hamz=\sum_{l=1}^{L}\haml$, where $\haml=\hat{I}\otimes\cdots\otimes\hat{I}\otimes \hat{H}_\mathrm{s}
\otimes\hat{I}\otimes\cdots\otimes\hat{I}$ and the driving Hamiltonian is of form~\eqref{eq:interaction_Hamiltonian}, the following inequality holds:
\begin{equation}\label{eq:corollary1.1}
\abs{P(t)}\leq k\norm{\hat{H}_\mathrm{s}-E^{(\mathrm{s})}_{\min}}\norm{\hamd(t)-v_{\min}(t)}/2, 
\end{equation}
where $E^{(\mathrm{s})}_{\min}$ is the single cell ground state energy. (In the main text, we denoted $E_{\mathrm{s}\min}\equiv E^{(\mathrm{s})}_{\min}$ for better readability.)
\end{corollary}
\begin{proof}
% By the theorem~\eqref{theo:main_theo}, maximum power is determined by maximum energy jump $\Delta E$. We will obtain maximum $\Delta E$ among $\hamd$ given as eq.~\eqref{eq:interaction_Hamiltonian}.

The single site Hamiltonian has spectral decomposition as $\hamz_\mathrm{s}=\sum_{j=1}^{N_\mathrm{s}} E_{j}^{(\mathrm{s})}\pro{E_{j}^{(\mathrm{s})}}{E_{j}^{(\mathrm{s})}}$, where $E_{j}^{(\mathrm{s})}$ is ordered by increasing energy, and $\ket{E_{j}^{(\mathrm{s})}}$ is the single site energy eigenstate. The basis of the initial Hamiltonian is rewritten as 
\[\begin{split}
&\hamz=\sum_{\bj=(1,\cdots,1)}^{(N_\mathrm{s},\cdots,N_\mathrm{s})} E_{\bj}^{(\mathrm{s})}\pro{E_{\bj}^{(\mathrm{s})}}{E_{\bj}^{(\mathrm{s})}}
\end{split}
\]
where $\bj=(j_1,j_2,\cdots,j_L)$, $E_{\bj}^{(\mathrm{s})}=\sum_{l=1}^{L} E_{j_l}^{(\mathrm{s})}$, and $\ket{E_{\bj}^{(\mathrm{s})}}=\ket{E_{j_1}^{(\mathrm{s})}}\otimes\cdots\otimes\ket{E_{j_L}^{(\mathrm{s})}}$.

An element of the driving Hamiltonian given by eq.~\eqref{eq:interaction_Hamiltonian}, $\hamd_\bi$, is decomposed in the basis of the initial Hamiltonian as
\[
\hamd_\bi=\sum_{\bj,\bm=(1,\cdots,1)}^{(N_\mathrm{s},\cdots,N_\mathrm{s})} V_{\bi,\bj\bm}\pro{E_{\bj}^{(\mathrm{s})}}{E_{\bm}^{(\mathrm{s})}}.
\]

To understand this complicated expression, we illustrate the use of indexes as follows: $\bi$ denotes positions of all the sites which interact through $\hamd_\bi$ (for example $\bi=(1,2)$ means that only the first and the second site interact), while $\bj=(j_1,\dots,j_L)$ and $\bm=(m_1,\dots,m_L)$ denote energy eigenstates of each site (for example, for $L=4$, $\bm=(1,2,1,4)$ corresponds to tensor product of local energy eigenstates $\ket{E_{(1,2,1,4)}^{(\mathrm{s})}}=\ket{E_{1}^{(\mathrm{s})}}\ket{E_{2}^{(\mathrm{s})}}\ket{E_{1}^{(\mathrm{s})}}\ket{E_{4}^{(\mathrm{s})}}$). Then by definition (using additionally $\bj=(1,3,2,1)$),
\[\label{eq:exampleofVielement}
V_{(1,2),(1,3,2,1)(1,2,1,4)}=\bra{E_{(1,3,2,1)}^{(\mathrm{s})}}\hamd_{(1,2)}\ket{E_{(1,2,1,4)}^{(\mathrm{s})}}
\]

Many of the elements of $V_{\bi}$ are zero. To show that, let us assume that $l\notin \bi$. Then choosing $\hat{M}=\pro{E_{j_l}}{E_{j_l}}$, from eq.~\eqref{eq:klocality} we obtain
\[\label{eq:McommutatorVi}
0=\bra{E_\bj}[\hat{M}^{(l)},\hamd_\bi(t)]\ket{E_\bm}=\bra{E_\bj}\hat{V}_\bi\ket{E_\bm}(1-\delta_{j_l m_l})
\]
Thus, if $j_l\neq m_l$, then
\[\label{eq:Vijm}
V_{\bi,\bj\bm}\equiv\bra{E_\bj}\hat{V}_\bi\ket{E_\bm}=0.
\]
To summarize, for any $l\notin \bi$, if $j_l\neq m_l$, then $V_{\bi,\bj\bm}=0$. In other words, every element $V_{\bi,\bj\bm}$ that is non-zero must have the same indexes $m_l=j_l$ for every $l\not\in\bi$.

For example, element~\eqref{eq:exampleofVielement} is zero, because $3\notin \bi=(1,2)$, and $2=j_3\neq m_3=1$. Elements that are allowed to be non-zero are, for example, $V_{(1,2),(1,3,1,1)(1,2,1,1)}$, $V_{(1,2),(2,4,2,1)(1,2,2,1)}$, or $V_{(1,2),(3,3,3,3)(3,2,3,3)}$.

In order to prove the Corollary, we must show that 
for all $\bj$ and $\bm$ such that $\abs{E_{\bj}-E_{\bm}}>k\norm{\hamz_\mathrm{s}-E^{(\mathrm{s})}_{\min}}$ (where $E_j$ and $E_m$ are eigen energies of the full Hamiltonian $\hamz=\sum_l\haml$) implies
\[\label{eq:Vjm0}
V_{\bj\bm} = 0.
\]
If we manage to show that, because $\Delta E$ is the minimum of all such numbers (see eq.~\eqref{eq:Vdefinition}), it must be that $k \norm{\hamz_\mathrm{s}-E^{(\mathrm{s})}_{\min}}\geq \Delta E$. Thus, the Corollary will follow from Theorem~\ref{theo:main_theo}.

Because $\norm{\hamz_\mathrm{s}-E^{(\mathrm{s})}_{\min}}$ is the difference between the largest and the smallest eigenvalue of $\hamz_\mathrm{s}$, it takes more than $k$ sites to achieve energy difference $\abs{E_{\bj}-E_{\bm}}$ bigger than $k \norm{\hamz_\mathrm{s}-E^{(\mathrm{s})}_{\min}}$, meaning that at least $k+1$ indices (elements) in $\bj$ and $\bm$ must differ. To prove that mathematically, for a contradiction we assume that there are $k'\leq k$ cells in which $\bj$ and $\bm$ are different, while $\abs{E_{\bj}-E_{\bm}}>k \norm{\hamz_\mathrm{s}-E^{(\mathrm{s})}_{\min}}$. We have
\[\label{eq:upperboundonejem}
\begin{split}
\abs{E_\bj-E_\bm}&=\abs{\sum_{l=1}^LE_{j_l}^{(\mathrm{s})}-E_{m_l}^{(\mathrm{s})}}=\abs{\sum_{l=1}^{k'}E_{j_{\pi(l)}}^{(\mathrm{s})}-E_{m_{\pi(l)}}^{(\mathrm{s})}}\\
&\leq \sum_{l=1}^{k'}\abs{E_{j_{\pi(l)}}^{(\mathrm{s})}-E_{m_{\pi(l)}}^{(\mathrm{s})}}
\leq k'\norm{\hamz_{\mathrm{s}}-E^{(\mathrm{s})}_{\min}},
\end{split}
\]
where $\pi(l)$ is a relabeling of indices to include only those where $j_l$ and $m_l$ are different. This is in contradiction with our assumption. Thus, if $\abs{E_{\bj}-E_{\bm}}>k \norm{\hamz_\mathrm{s}-E^{(\mathrm{s})}_{\min}}$, then vectors $\bj$ and $\bm$ have at least $k+1$ different elements.

However, since vector $\bi$ in eq.~\eqref{eq:Vijm} has at most $k$ elements according to the assumption of the Corollary, but there are at least $k+1$ elements of $\bj$ and $\bm$ that differ, according to that equation it must be that
\[
V_{\bi,\bj\bm}=0.
\]
This holds for every $\bi$ that make up $\hamd$, thus also
\[
V_{\bj\bm}=0,
\]
which, by explanation below eq.~\eqref{eq:Vjm0}, proves the Corollary.

\end{proof}

Finally, we prove eq.~(17) in the main text.

\begin{corollary}
Let us consider identical cells, which have initial Hamiltonian $\hamz=\sum_{l=1}^{L}\haml$, where $\haml=\hat{I}\otimes\cdots\otimes\hat{I}\otimes \hat{H}_\mathrm{s} \otimes\hat{I}\otimes\cdots\otimes\hat{I}$. Where we rewrite driving Hamiltonian $\hamd=\sum_{\bi\in K(L,k)}{\hat{V}_{\bi}}$ as
\[
\hamd=\sum_{k=1}^{N}\sum_{\bi\in C(L,k)}{\hat{V}_{\bi}}=\sum_{k=1}^{N}\hat{V}_{k}.
\]
The following inequality holds:
\[
\abs{P(t)}\leq \sum_{k=1}^{L} k\norm{\hamd_{k}-v_{k,\min}}\norm{\hat{H}_\mathrm{s}-E^{(\mathrm{s})}_{\min}}/2.
\]

% S-number sequence $\Delta E^s$ is ordered, \textit{i.e.}  $\Delta E^{s+1}>\Delta E^s$, where $\Delta E^S=\Delta E$ and $\Delta E^1\geq 0$. We define $\hamd^k$ as 
% \[
% \begin{split}
% \hamd^s&=\sum_{j , m=1}^{N} V_{jm}^{k}\ket{E_{j}}\bra{E_m}\\
% V_{jm}^{s}&=
% \begin{cases}
% V_{jm}\quad &\mathrm{for}\ \Delta E^{(s-1)}<|E_j-E_m|\leq \Delta E^s\\
% 0\quad &\mathrm{for\ other\ cases} 
% \end{cases}
% \end{split}
% \]
% for $1<s\leq S$ and
% \[
% \begin{split}
% V_{jm}^{s}&=
% \begin{cases}
% V_{jm}\quad &\mathrm{for}\ \Delta 0\leq |E_j-E_m|\leq \Delta E^1\\
% 0\quad &\mathrm{for} \Delta E^1\leq |E_j-E_m| 
% \end{cases}
% \end{split}
% \]
% for $s=1$.
% Then the following inequality holds
% \[
% \norm{[\hamz,\hamd]}\leq \sum_{s=1}^{S}\Delta E^s \norm{\hamd^s}
% \]
\end{corollary}
\begin{proof}
By triangular inequality, the charging power is bounded as
\[
\abs{P(t)}\leq\norm{[\hamz,\hamd]}\leq\sum_{k=1}^{N}\norm{[\hamz,\hamd_{k}]}.
\]
Using Corollary~\ref{corol:2}, we obtain
\[\begin{split}
&\abs{P(t)}\leq\sum_{k=1}^{N}\norm{[\hamz,\hamd_{k}]}\\
&\leq\sum_{k=1}^{L} k\norm{\hamd_{k}-v_{k,\min}}\norm{\hat{H}_\mathrm{s}-E^{(\mathrm{s})}_{\min}}/2.
\end{split}
\]

% By definition, driving Hamiltonian is expressed as sum of $\hamd^s$, \textit{i.e.} $\hamd=\sum_{s=1}^{S}\hamd^s$. $\norm{[\hamz,\hamd]}$ is rewritten as
% \[
% \norm{[\hamz,\hamd]}=\sum_{s=1}^{S}[\hamz,\hamd^s]
% \]
% By theorem~\ref{theo:main_theo} we obtain
% \[
% \norm{[\hamz,\hamd]}\leq \sum_{s=1}^{S}\Delta E^s \norm{\hamd^s}
% \]
% which conclude proof.
\end{proof}

\section{\label{sec:level3} Random all-to-all Hamiltonian}

In this section, we estimate the normalization factor in the SY-like Hamiltonian used as an example in the main text, and we provide an explanation of why in Fig 1.~(b) the maximum decreases slightly with growing size of the system $L$.

Let us consider a battery made of a random, $2$-local, \textit{all-to-all} driving Hamiltonian given by 
\[\label{eq:driving_Hamiltonian_random}
\hamd=C\, \hamd_0,
\]
where
\[\label{eq:driving_Hamiltonian_random0}
\hamd_0= \sum_{i < j}^L \sum_{\alpha = x,\, y,\, z} J_{ij}^{\alpha} \hat{\sigma}_i^\alpha \hat{\sigma}_j^\alpha.
\]
the couplings $J_{ij}^{\alpha}$ are randomly extracted from a normal distribution and $C=\frac{2}{\norm{\hamd_0-v_{0,\min}}}$ is a normalization factor which ensures that $\norm{\hamd-v_{\min}}=2$. We consider the initial Hamiltonian to be $\hamz=\sum_{i=1}^{L} h\hat{\sigma}_i^z$. 

Since $C$ depends on $J_{ij}^{\alpha}$, which are randomly extracted from a normal distribution, we see that $C$ as well changes randomly in each simulation. 
Our goal is to estimate the scaling of the quantity $\norm{\hamd_0-v_{0,\min}}$ with the system size, $L$.

To this end, we start by recalling that, for large $L$, the eigenvalues of $\hamd_0$, $v_{0, i}$, follow a Gaussian distribution. 
In particular, they satisfy
\[
\mean{\sum_{i=1}^{2^L} {v_{0, i}}}_J=\mean{\mathrm{Tr}({\hamd_0})}_J=0,
\]
as well as
\[\label{eq:sigma}
\begin{split}
&\sigma^22^L=:\mean{\sum_{i=1}^{2^L} {v_{0, i}}^2}_J=\mean{\mathrm{Tr}( {\hamd_0}^2)}_J\\
&=\mean{\mathrm{Tr}\left(\sum_{i < j}^L \sum_{\alpha = x,\, y,\, z}\sum_{i' < j'}^L \sum_{\alpha' = x,\, y,\, z} J_{ij}^{\alpha}J_{i'j'}^{\alpha'} \hat{\sigma}_i^\alpha \hat{\sigma}_j^\alpha \hat{\sigma}_{i'}^{\alpha'} \hat{\sigma}_{j'}^{\alpha'}\right)}_J,
\end{split}
\]
where $\langle\rangle_J$ denotes ensemble averaging over $J_{ij}^{\alpha}$. Since trace of $\sigma^{\alpha}$ is $0$, trace of elements in \eqref{eq:sigma} are $0$, except in the case $i=i'$, $j=j'$ and $\alpha=\alpha'$. 
Using that the variance of $J_{ij}^{\alpha}$ is $1$, we obtain
\[
\sigma^22^L=\mean{\mathrm{Tr}\left(\sum_{i<j}^L \sum_{\alpha = x,\, y,\, z} (J_{ij}^{\alpha})^2 \mathds{1}\right)}_J
=\frac{3L(L-1)}{2}2^{L}.
\]
Hence, the eigenvalues of $\hamd_0$, $v_{0, i}$, follow a Gaussian distribution having zero mean and variance given by $\sigma^2=\frac{3L(L-1)}{2}$.

As last step, we need to estimate the expectation value for the maximum among $N$ numbers randomly extracted from a Gaussian distribution having vanishing mean value and variance equal to $\sigma$.
It can be shown, \cite{kamath2015bounds}, that such a maximum value goes like  $\sigma\sqrt{2\ln{N}}$ for large $N$. 
As a result, using that, in our case, $N=2^L$ the mean value of $\frac{\norm{\hamd_0-v_{0,\min}}}{2}$averaged over $J_{ij}^{\alpha}$ takes the form
\[\label{eq:analytical_estimate}
\begin{split}
&\mean{\frac{\norm{\hamd_0-v_{0,\min}}}{2}}_J=\mean{\frac{v_{0,\max}-v_{0,\min}}{2}}_J=\mean{v_{0,\max}}_J\\
&=\sigma\sqrt{2\ln(2^L)}=\sqrt{\frac{3L(L-1)}{2} 2 L \ln 2}=\sqrt{3L^2(L-1)\ln 2}
\end{split}
\]
\begin{figure}[t!]
\begin{center}
    \includegraphics[width=.9\hsize]{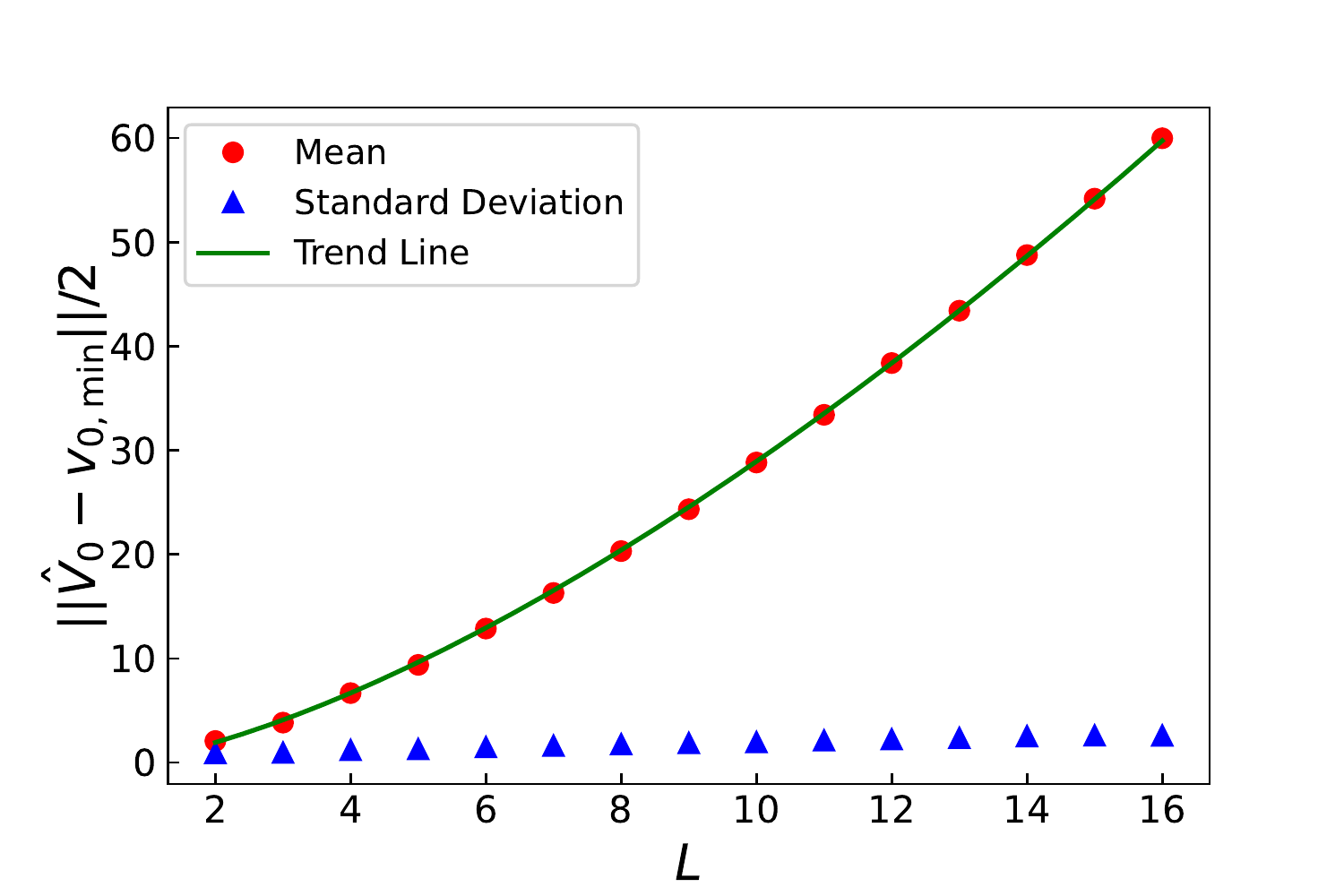}
    \caption{
    Mean and standard deviation of $\norm{\hamd_0-v_{0\min}}/2$ for each $L$. Total number of realization is 10000. The numerical data are fitted against an Ansatz of the form $\alpha\sqrt{L^2(L-1)}$. The fitting parameter, $\alpha$, takes the value $\alpha = 0.964111$.
    }
\label{fig:figure1}
\end{center}
\end{figure}
Motivated by this result, in Fig.~\ref{fig:figure2} we numerically computed $\mean{\frac{\norm{\hamd_0-v_{0,\min}}}{2}}_J$ for $L$ up to $16$ and we have fitted the resulting values against an Ansatz of the form $\alpha \sqrt{L^2 (L-1)}$, with $\alpha$ being a fitting parameter.
As it is evident from the figure the agreement is excellent. Thus, from the analytical estimate, Eq.~\eqref{eq:analytical_estimate}, tested in Fig.~\ref{fig:figure1}, we conclude that the normalization factor scales as $C\propto \sqrt{L^2(L-1)}$.

\begin{figure}[t!]
\begin{center}
    \includegraphics[width=.9\hsize]{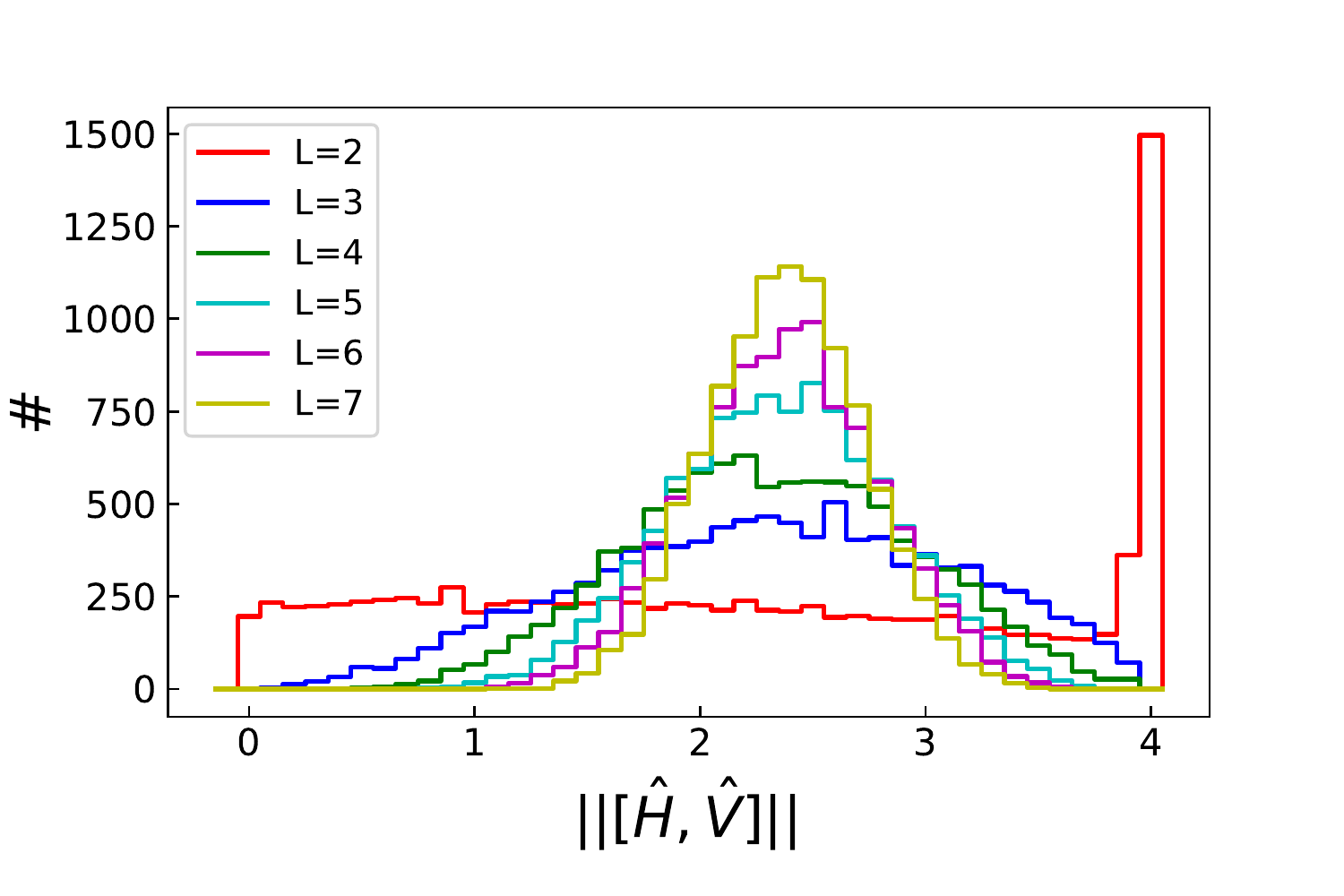}\\
    \ \ (a)\\
    \includegraphics[width=.9\hsize]{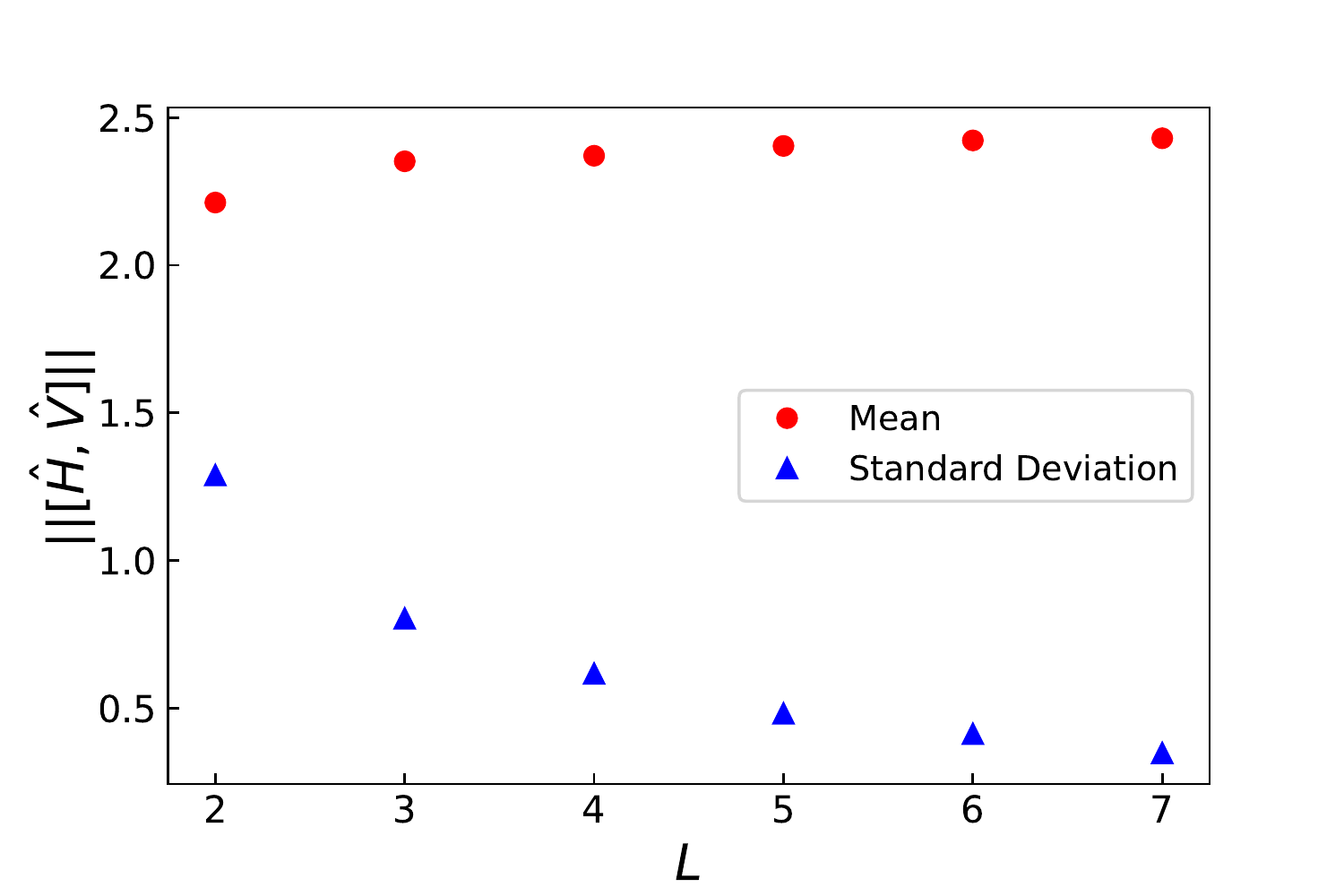}\\
    \ \ (b)
    \caption{
    (a) Number of the cases $x+0.1>\norm{[\hamd-\hamz]}\geq x$ for each $L$. Total number of realization is 10000.
    (b) Mean and standard deviation of $\norm{[\hamd-\hamz]}$ for each $L$. Total number of realization is 10000.
    }
\label{fig:figure2}
\end{center}
\end{figure}

Now we will explain why in Fig 1.~(b) the maximum decreases slightly with growing size of the system $L$. Also the quantity $\norm{[\hamz,\hamd]}$ is determined in terms $J_{ij}^{\alpha}$; therefore, it also changes in each realization (instance) of the simulation. 
Fig.~\ref{fig:figure2} (a) shows the distribution (histogram) of $\norm{[\hamz,\hamd]}$ as a function of $L$. As $L$ increase, the shape of distribution becomes sharper and localized around the mean value $\sim 2.5$. Since Fig 1.~(b) in the main text shows the maxima from a finite number of 500 realizations, the maximum drawn from this finite number also decreases. Additionally, from Fig.~\ref{fig:figure2} (b) we see that the standard deviation decreases faster with $L$ than the raising of the mean value. To conclude, the probability of getting large values of $\norm{[\hamz,\hamd]}$ decreases with $L$.  For this reason, the expectation value of the maximum value of $\norm{[\hamz,\hamd]}$ reduces with $L$, if the number of realizations stays fixed for every $L$. Since $\norm{[\hamz,\hamd]}$ is also an upper bound on maximum charging power $P_{\max}$, we expect a similar behavior also for this quantity.

% https://stats.stackexchange.com/questions/343914/expected-value-of-maximum-of-samples-from-normal-distribution
% https://en.wikipedia.org/wiki/Normal_distribution#Quantile_function



\begin{thebibliography}{46}%
    \makeatletter
    \providecommand \@ifxundefined [1]{%
     \@ifx{#1\undefined}
    }%
    \providecommand \@ifnum [1]{%
     \ifnum #1\expandafter \@firstoftwo
     \else \expandafter \@secondoftwo
     \fi
    }%
    \providecommand \@ifx [1]{%
     \ifx #1\expandafter \@firstoftwo
     \else \expandafter \@secondoftwo
     \fi
    }%
    \providecommand \natexlab [1]{#1}%
    \providecommand \enquote  [1]{``#1''}%
    \providecommand \bibnamefont  [1]{#1}%
    \providecommand \bibfnamefont [1]{#1}%
    \providecommand \citenamefont [1]{#1}%
    \providecommand \href@noop [0]{\@secondoftwo}%
    \providecommand \href [0]{\begingroup \@sanitize@url \@href}%
    \providecommand \@href[1]{\@@startlink{#1}\@@href}%
    \providecommand \@@href[1]{\endgroup#1\@@endlink}%
    \providecommand \@sanitize@url [0]{\catcode `\\12\catcode `\$12\catcode
      `\&12\catcode `\#12\catcode `\^12\catcode `\_12\catcode `\%12\relax}%
    \providecommand \@@startlink[1]{}%
    \providecommand \@@endlink[0]{}%
    \providecommand \url  [0]{\begingroup\@sanitize@url \@url }%
    \providecommand \@url [1]{\endgroup\@href {#1}{\urlprefix }}%
    \providecommand \urlprefix  [0]{URL }%
    \providecommand \Eprint [0]{\href }%
    \providecommand \doibase [0]{http://dx.doi.org/}%
    \providecommand \selectlanguage [0]{\@gobble}%
    \providecommand \bibinfo  [0]{\@secondoftwo}%
    \providecommand \bibfield  [0]{\@secondoftwo}%
    \providecommand \translation [1]{[#1]}%
    \providecommand \BibitemOpen [0]{}%
    \providecommand \bibitemStop [0]{}%
    \providecommand \bibitemNoStop [0]{.\EOS\space}%
    \providecommand \EOS [0]{\spacefactor3000\relax}%
    \providecommand \BibitemShut  [1]{\csname bibitem#1\endcsname}%
    \let\auto@bib@innerbib\@empty
    %</preamble>
    \bibitem [{\citenamefont {{Giovannetti}}\ \emph {et~al.}(2011)\citenamefont
      {{Giovannetti}}, \citenamefont {{Lloyd}},\ and\ \citenamefont
      {{Maccone}}}]{giovannetti2011advances}%
      \BibitemOpen
      \bibfield  {author} {\bibinfo {author} {\bibfnamefont {V.}~\bibnamefont
      {{Giovannetti}}}, \bibinfo {author} {\bibfnamefont {S.}~\bibnamefont
      {{Lloyd}}}, \ and\ \bibinfo {author} {\bibfnamefont {L.}~\bibnamefont
      {{Maccone}}},\ }\href {\doibase 10.1038/nphoton.2011.35} {\bibfield
      {journal} {\bibinfo  {journal} {Nature Photonics}\ }\textbf {\bibinfo
      {volume} {5}},\ \bibinfo {pages} {222} (\bibinfo {year} {2011})},\ \Eprint
      {http://arxiv.org/abs/1102.2318} {arXiv:1102.2318 [quant-ph]} \BibitemShut
      {NoStop}%
    \bibitem [{\citenamefont {Abbott}(2016)}]{abbott2016observation}%
      \BibitemOpen
      \bibfield  {author} {\bibinfo {author} {\bibfnamefont {B.~P. e.~A.}\
      \bibnamefont {Abbott}} (\bibinfo {collaboration} {LIGO Scientific
      Collaboration and Virgo Collaboration}),\ }\href {\doibase
      10.1103/PhysRevLett.116.061102} {\bibfield  {journal} {\bibinfo  {journal}
      {Phys. Rev. Lett.}\ }\textbf {\bibinfo {volume} {116}},\ \bibinfo {pages}
      {061102} (\bibinfo {year} {2016})}\BibitemShut {NoStop}%
    \bibitem [{\citenamefont {{Pirandola}}\ \emph {et~al.}(2020)\citenamefont
      {{Pirandola}}, \citenamefont {{Andersen}}, \citenamefont {{Banchi}},
      \citenamefont {{Berta}}, \citenamefont {{Bunandar}}, \citenamefont
      {{Colbeck}}, \citenamefont {{Englund}}, \citenamefont {{Gehring}},
      \citenamefont {{Lupo}}, \citenamefont {{Ottaviani}}, \citenamefont
      {{Pereira}}, \citenamefont {{Razavi}}, \citenamefont {{Shamsul Shaari}},
      \citenamefont {{Tomamichel}}, \citenamefont {{Usenko}}, \citenamefont
      {{Vallone}}, \citenamefont {{Villoresi}},\ and\ \citenamefont
      {{Wallden}}}]{pirandola2020advances}%
      \BibitemOpen
      \bibfield  {author} {\bibinfo {author} {\bibfnamefont {S.}~\bibnamefont
      {{Pirandola}}}, \bibinfo {author} {\bibfnamefont {U.~L.}\ \bibnamefont
      {{Andersen}}}, \bibinfo {author} {\bibfnamefont {L.}~\bibnamefont
      {{Banchi}}}, \bibinfo {author} {\bibfnamefont {M.}~\bibnamefont {{Berta}}},
      \bibinfo {author} {\bibfnamefont {D.}~\bibnamefont {{Bunandar}}}, \bibinfo
      {author} {\bibfnamefont {R.}~\bibnamefont {{Colbeck}}}, \bibinfo {author}
      {\bibfnamefont {D.}~\bibnamefont {{Englund}}}, \bibinfo {author}
      {\bibfnamefont {T.}~\bibnamefont {{Gehring}}}, \bibinfo {author}
      {\bibfnamefont {C.}~\bibnamefont {{Lupo}}}, \bibinfo {author} {\bibfnamefont
      {C.}~\bibnamefont {{Ottaviani}}}, \bibinfo {author} {\bibfnamefont {J.~L.}\
      \bibnamefont {{Pereira}}}, \bibinfo {author} {\bibfnamefont {M.}~\bibnamefont
      {{Razavi}}}, \bibinfo {author} {\bibfnamefont {J.}~\bibnamefont {{Shamsul
      Shaari}}}, \bibinfo {author} {\bibfnamefont {M.}~\bibnamefont
      {{Tomamichel}}}, \bibinfo {author} {\bibfnamefont {V.~C.}\ \bibnamefont
      {{Usenko}}}, \bibinfo {author} {\bibfnamefont {G.}~\bibnamefont {{Vallone}}},
      \bibinfo {author} {\bibfnamefont {P.}~\bibnamefont {{Villoresi}}}, \ and\
      \bibinfo {author} {\bibfnamefont {P.}~\bibnamefont {{Wallden}}},\ }\href
      {\doibase 10.1364/AOP.361502} {\bibfield  {journal} {\bibinfo  {journal}
      {Advances in Optics and Photonics}\ }\textbf {\bibinfo {volume} {12}},\
      \bibinfo {pages} {1012} (\bibinfo {year} {2020})},\ \Eprint
      {http://arxiv.org/abs/1906.01645} {arXiv:1906.01645 [quant-ph]} \BibitemShut
      {NoStop}%
    \bibitem [{\citenamefont {{Elliott}}\ \emph {et~al.}(2005)\citenamefont
      {{Elliott}}, \citenamefont {{Colvin}}, \citenamefont {{Pearson}},
      \citenamefont {{Pikalo}}, \citenamefont {{Schlafer}},\ and\ \citenamefont
      {{Yeh}}}]{elliott2005current}%
      \BibitemOpen
      \bibfield  {author} {\bibinfo {author} {\bibfnamefont {C.}~\bibnamefont
      {{Elliott}}}, \bibinfo {author} {\bibfnamefont {A.}~\bibnamefont {{Colvin}}},
      \bibinfo {author} {\bibfnamefont {D.}~\bibnamefont {{Pearson}}}, \bibinfo
      {author} {\bibfnamefont {O.}~\bibnamefont {{Pikalo}}}, \bibinfo {author}
      {\bibfnamefont {J.}~\bibnamefont {{Schlafer}}}, \ and\ \bibinfo {author}
      {\bibfnamefont {H.}~\bibnamefont {{Yeh}}},\ }\href@noop {} {\bibfield
      {journal} {\bibinfo  {journal} {arXiv e-prints}\ ,\ \bibinfo {eid}
      {quant-ph/0503058}} (\bibinfo {year} {2005})},\ \Eprint
      {http://arxiv.org/abs/quant-ph/0503058} {arXiv:quant-ph/0503058 [quant-ph]}
      \BibitemShut {NoStop}%
    \bibitem [{\citenamefont {Chen}\ \emph {et~al.}(2021)\citenamefont {Chen},
      \citenamefont {Zhang}, \citenamefont {Chen}, \citenamefont {Cai},
      \citenamefont {Liao}, \citenamefont {Zhang}, \citenamefont {Chen},
      \citenamefont {Yin}, \citenamefont {Ren}, \citenamefont {Chen} \emph
      {et~al.}}]{chen2021integrated}%
      \BibitemOpen
      \bibfield  {author} {\bibinfo {author} {\bibfnamefont {Y.-A.}\ \bibnamefont
      {Chen}}, \bibinfo {author} {\bibfnamefont {Q.}~\bibnamefont {Zhang}},
      \bibinfo {author} {\bibfnamefont {T.-Y.}\ \bibnamefont {Chen}}, \bibinfo
      {author} {\bibfnamefont {W.-Q.}\ \bibnamefont {Cai}}, \bibinfo {author}
      {\bibfnamefont {S.-K.}\ \bibnamefont {Liao}}, \bibinfo {author}
      {\bibfnamefont {J.}~\bibnamefont {Zhang}}, \bibinfo {author} {\bibfnamefont
      {K.}~\bibnamefont {Chen}}, \bibinfo {author} {\bibfnamefont {J.}~\bibnamefont
      {Yin}}, \bibinfo {author} {\bibfnamefont {J.-G.}\ \bibnamefont {Ren}},
      \bibinfo {author} {\bibfnamefont {Z.}~\bibnamefont {Chen}},  \emph {et~al.},\
      }\href {https://www.nature.com/articles/s41586-020-03093-8} {\bibfield
      {journal} {\bibinfo  {journal} {Nature}\ }\textbf {\bibinfo {volume} {589}},\
      \bibinfo {pages} {214} (\bibinfo {year} {2021})}\BibitemShut {NoStop}%
    \bibitem [{\citenamefont {{Cao}}\ \emph {et~al.}(2018)\citenamefont {{Cao}},
      \citenamefont {{Romero}}, \citenamefont {{Olson}}, \citenamefont
      {{Degroote}}, \citenamefont {{Johnson}}, \citenamefont {{Kieferov{\'a}}},
      \citenamefont {{Kivlichan}}, \citenamefont {{Menke}}, \citenamefont
      {{Peropadre}}, \citenamefont {{Sawaya}}, \citenamefont {{Sim}}, \citenamefont
      {{Veis}},\ and\ \citenamefont {{Aspuru-Guzik}}}]{cao2018quantum}%
      \BibitemOpen
      \bibfield  {author} {\bibinfo {author} {\bibfnamefont {Y.}~\bibnamefont
      {{Cao}}}, \bibinfo {author} {\bibfnamefont {J.}~\bibnamefont {{Romero}}},
      \bibinfo {author} {\bibfnamefont {J.~P.}\ \bibnamefont {{Olson}}}, \bibinfo
      {author} {\bibfnamefont {M.}~\bibnamefont {{Degroote}}}, \bibinfo {author}
      {\bibfnamefont {P.~D.}\ \bibnamefont {{Johnson}}}, \bibinfo {author}
      {\bibfnamefont {M.}~\bibnamefont {{Kieferov{\'a}}}}, \bibinfo {author}
      {\bibfnamefont {I.~D.}\ \bibnamefont {{Kivlichan}}}, \bibinfo {author}
      {\bibfnamefont {T.}~\bibnamefont {{Menke}}}, \bibinfo {author} {\bibfnamefont
      {B.}~\bibnamefont {{Peropadre}}}, \bibinfo {author} {\bibfnamefont
      {N.~P.~D.}\ \bibnamefont {{Sawaya}}}, \bibinfo {author} {\bibfnamefont
      {S.}~\bibnamefont {{Sim}}}, \bibinfo {author} {\bibfnamefont
      {L.}~\bibnamefont {{Veis}}}, \ and\ \bibinfo {author} {\bibfnamefont
      {A.}~\bibnamefont {{Aspuru-Guzik}}},\ }\href@noop {} {\bibfield  {journal}
      {\bibinfo  {journal} {arXiv e-prints}\ ,\ \bibinfo {eid} {arXiv:1812.09976}}
      (\bibinfo {year} {2018})},\ \Eprint {http://arxiv.org/abs/1812.09976}
      {arXiv:1812.09976 [quant-ph]} \BibitemShut {NoStop}%
    \bibitem [{\citenamefont {Aaronson}(2008)}]{aaronson2008limits}%
      \BibitemOpen
      \bibfield  {author} {\bibinfo {author} {\bibfnamefont {S.}~\bibnamefont
      {Aaronson}},\ }\href
      {https://www.scientificamerican.com/article/the-limits-of-quantum-computers/}
      {\bibfield  {journal} {\bibinfo  {journal} {Scientific American}\ }\textbf
      {\bibinfo {volume} {298}},\ \bibinfo {pages} {62} (\bibinfo {year}
      {2008})}\BibitemShut {NoStop}%
    \bibitem [{\citenamefont {{Menges}}\ \emph {et~al.}(2016)\citenamefont
      {{Menges}}, \citenamefont {{Mensch}}, \citenamefont {{Schmid}}, \citenamefont
      {{Riel}}, \citenamefont {{Stemmer}},\ and\ \citenamefont
      {{Gotsmann}}}]{menges2016temperature}%
      \BibitemOpen
      \bibfield  {author} {\bibinfo {author} {\bibfnamefont {F.}~\bibnamefont
      {{Menges}}}, \bibinfo {author} {\bibfnamefont {P.}~\bibnamefont {{Mensch}}},
      \bibinfo {author} {\bibfnamefont {H.}~\bibnamefont {{Schmid}}}, \bibinfo
      {author} {\bibfnamefont {H.}~\bibnamefont {{Riel}}}, \bibinfo {author}
      {\bibfnamefont {A.}~\bibnamefont {{Stemmer}}}, \ and\ \bibinfo {author}
      {\bibfnamefont {B.}~\bibnamefont {{Gotsmann}}},\ }\href {\doibase
      10.1038/ncomms10874} {\bibfield  {journal} {\bibinfo  {journal} {Nature
      Communications}\ }\textbf {\bibinfo {volume} {7}},\ \bibinfo {eid} {10874}
      (\bibinfo {year} {2016})}\BibitemShut {NoStop}%
    \bibitem [{\citenamefont {Campaioli}\ \emph {et~al.}(2018)\citenamefont
      {Campaioli}, \citenamefont {Pollock},\ and\ \citenamefont
      {Vinjanampathy}}]{campaioli2018quantum}%
      \BibitemOpen
      \bibfield  {author} {\bibinfo {author} {\bibfnamefont {F.}~\bibnamefont
      {Campaioli}}, \bibinfo {author} {\bibfnamefont {F.~A.}\ \bibnamefont
      {Pollock}}, \ and\ \bibinfo {author} {\bibfnamefont {S.}~\bibnamefont
      {Vinjanampathy}},\ }\href@noop {} {\enquote {\bibinfo {title} {Quantum
      batteries - review chapter},}\ } (\bibinfo {year} {2018}),\ \Eprint
      {http://arxiv.org/abs/1805.05507} {arXiv:1805.05507 [quant-ph]} \BibitemShut
      {NoStop}%
    \bibitem [{\citenamefont {Bhattacharjee}\ and\ \citenamefont
      {Dutta}(2020)}]{bhattacharjee2020quantum}%
      \BibitemOpen
      \bibfield  {author} {\bibinfo {author} {\bibfnamefont {S.}~\bibnamefont
      {Bhattacharjee}}\ and\ \bibinfo {author} {\bibfnamefont {A.}~\bibnamefont
      {Dutta}},\ }\href@noop {} {\enquote {\bibinfo {title} {Quantum thermal
      machines and batteries},}\ } (\bibinfo {year} {2020}),\ \Eprint
      {http://arxiv.org/abs/2008.07889} {arXiv:2008.07889 [quant-ph]} \BibitemShut
      {NoStop}%
    \bibitem [{\citenamefont {Alicki}\ and\ \citenamefont
      {Fannes}(2013)}]{PhysRevE.87.042123}%
      \BibitemOpen
      \bibfield  {author} {\bibinfo {author} {\bibfnamefont {R.}~\bibnamefont
      {Alicki}}\ and\ \bibinfo {author} {\bibfnamefont {M.}~\bibnamefont
      {Fannes}},\ }\href {\doibase 10.1103/PhysRevE.87.042123} {\bibfield
      {journal} {\bibinfo  {journal} {Phys. Rev. E}\ }\textbf {\bibinfo {volume}
      {87}},\ \bibinfo {pages} {042123} (\bibinfo {year} {2013})}\BibitemShut
      {NoStop}%
    \bibitem [{\citenamefont {Hovhannisyan}\ \emph {et~al.}(2013)\citenamefont
      {Hovhannisyan}, \citenamefont {Perarnau-Llobet}, \citenamefont {Huber},\ and\
      \citenamefont {Acín}}]{Hovhannisyan_2013}%
      \BibitemOpen
      \bibfield  {author} {\bibinfo {author} {\bibfnamefont {K.~V.}\ \bibnamefont
      {Hovhannisyan}}, \bibinfo {author} {\bibfnamefont {M.}~\bibnamefont
      {Perarnau-Llobet}}, \bibinfo {author} {\bibfnamefont {M.}~\bibnamefont
      {Huber}}, \ and\ \bibinfo {author} {\bibfnamefont {A.}~\bibnamefont
      {Acín}},\ }\href {\doibase 10.1103/physrevlett.111.240401} {\bibfield
      {journal} {\bibinfo  {journal} {Physical Review Letters}\ }\textbf {\bibinfo
      {volume} {111}} (\bibinfo {year} {2013}),\
      10.1103/physrevlett.111.240401}\BibitemShut {NoStop}%
    \bibitem [{\citenamefont {Andolina}\ \emph {et~al.}(2018)\citenamefont
      {Andolina}, \citenamefont {Farina}, \citenamefont {Mari}, \citenamefont
      {Pellegrini}, \citenamefont {Giovannetti},\ and\ \citenamefont
      {Polini}}]{PhysRevB.98.205423}%
      \BibitemOpen
      \bibfield  {author} {\bibinfo {author} {\bibfnamefont {G.~M.}\ \bibnamefont
      {Andolina}}, \bibinfo {author} {\bibfnamefont {D.}~\bibnamefont {Farina}},
      \bibinfo {author} {\bibfnamefont {A.}~\bibnamefont {Mari}}, \bibinfo {author}
      {\bibfnamefont {V.}~\bibnamefont {Pellegrini}}, \bibinfo {author}
      {\bibfnamefont {V.}~\bibnamefont {Giovannetti}}, \ and\ \bibinfo {author}
      {\bibfnamefont {M.}~\bibnamefont {Polini}},\ }\href {\doibase
      10.1103/PhysRevB.98.205423} {\bibfield  {journal} {\bibinfo  {journal} {Phys.
      Rev. B}\ }\textbf {\bibinfo {volume} {98}},\ \bibinfo {pages} {205423}
      (\bibinfo {year} {2018})}\BibitemShut {NoStop}%
    \bibitem [{\citenamefont {Zhang}\ \emph {et~al.}(2019)\citenamefont {Zhang},
      \citenamefont {Yang}, \citenamefont {Fu},\ and\ \citenamefont
      {Wang}}]{PhysRevE.99.052106}%
      \BibitemOpen
      \bibfield  {author} {\bibinfo {author} {\bibfnamefont {Y.-Y.}\ \bibnamefont
      {Zhang}}, \bibinfo {author} {\bibfnamefont {T.-R.}\ \bibnamefont {Yang}},
      \bibinfo {author} {\bibfnamefont {L.}~\bibnamefont {Fu}}, \ and\ \bibinfo
      {author} {\bibfnamefont {X.}~\bibnamefont {Wang}},\ }\href {\doibase
      10.1103/PhysRevE.99.052106} {\bibfield  {journal} {\bibinfo  {journal} {Phys.
      Rev. E}\ }\textbf {\bibinfo {volume} {99}},\ \bibinfo {pages} {052106}
      (\bibinfo {year} {2019})}\BibitemShut {NoStop}%
    \bibitem [{\citenamefont {Caravelli}\ \emph {et~al.}(2020)\citenamefont
      {Caravelli}, \citenamefont {Coulter-De~Wit}, \citenamefont
      {Garc\'{\i}a-Pintos},\ and\ \citenamefont
      {Hamma}}]{PhysRevResearch.2.023095}%
      \BibitemOpen
      \bibfield  {author} {\bibinfo {author} {\bibfnamefont {F.}~\bibnamefont
      {Caravelli}}, \bibinfo {author} {\bibfnamefont {G.}~\bibnamefont
      {Coulter-De~Wit}}, \bibinfo {author} {\bibfnamefont {L.~P.}\ \bibnamefont
      {Garc\'{\i}a-Pintos}}, \ and\ \bibinfo {author} {\bibfnamefont
      {A.}~\bibnamefont {Hamma}},\ }\href {\doibase
      10.1103/PhysRevResearch.2.023095} {\bibfield  {journal} {\bibinfo  {journal}
      {Phys. Rev. Research}\ }\textbf {\bibinfo {volume} {2}},\ \bibinfo {pages}
      {023095} (\bibinfo {year} {2020})}\BibitemShut {NoStop}%
    \bibitem [{\citenamefont {Quach}\ and\ \citenamefont
      {Munro}(2020)}]{Quach_2020}%
      \BibitemOpen
      \bibfield  {author} {\bibinfo {author} {\bibfnamefont {J.~Q.}\ \bibnamefont
      {Quach}}\ and\ \bibinfo {author} {\bibfnamefont {W.~J.}\ \bibnamefont
      {Munro}},\ }\href {\doibase 10.1103/physrevapplied.14.024092} {\bibfield
      {journal} {\bibinfo  {journal} {Physical Review Applied}\ }\textbf {\bibinfo
      {volume} {14}} (\bibinfo {year} {2020}),\
      10.1103/physrevapplied.14.024092}\BibitemShut {NoStop}%
    \bibitem [{\citenamefont {Crescente}\ \emph
      {et~al.}(2020{\natexlab{a}})\citenamefont {Crescente}, \citenamefont
      {Carrega}, \citenamefont {Sassetti},\ and\ \citenamefont
      {Ferraro}}]{Crescente_2020fluct}%
      \BibitemOpen
      \bibfield  {author} {\bibinfo {author} {\bibfnamefont {A.}~\bibnamefont
      {Crescente}}, \bibinfo {author} {\bibfnamefont {M.}~\bibnamefont {Carrega}},
      \bibinfo {author} {\bibfnamefont {M.}~\bibnamefont {Sassetti}}, \ and\
      \bibinfo {author} {\bibfnamefont {D.}~\bibnamefont {Ferraro}},\ }\href
      {\doibase 10.1088/1367-2630/ab91fc} {\bibfield  {journal} {\bibinfo
      {journal} {New Journal of Physics}\ }\textbf {\bibinfo {volume} {22}},\
      \bibinfo {pages} {063057} (\bibinfo {year} {2020}{\natexlab{a}})}\BibitemShut
      {NoStop}%
    \bibitem [{\citenamefont {Friis}\ and\ \citenamefont
      {Huber}(2018)}]{Friis_2018}%
      \BibitemOpen
      \bibfield  {author} {\bibinfo {author} {\bibfnamefont {N.}~\bibnamefont
      {Friis}}\ and\ \bibinfo {author} {\bibfnamefont {M.}~\bibnamefont {Huber}},\
      }\href {\doibase 10.22331/q-2018-04-23-61} {\bibfield  {journal} {\bibinfo
      {journal} {Quantum}\ }\textbf {\bibinfo {volume} {2}},\ \bibinfo {pages} {61}
      (\bibinfo {year} {2018})}\BibitemShut {NoStop}%
    \bibitem [{\citenamefont {Rossini}\ \emph {et~al.}(2019)\citenamefont
      {Rossini}, \citenamefont {Andolina},\ and\ \citenamefont
      {Polini}}]{PhysRevB.100.115142}%
      \BibitemOpen
      \bibfield  {author} {\bibinfo {author} {\bibfnamefont {D.}~\bibnamefont
      {Rossini}}, \bibinfo {author} {\bibfnamefont {G.~M.}\ \bibnamefont
      {Andolina}}, \ and\ \bibinfo {author} {\bibfnamefont {M.}~\bibnamefont
      {Polini}},\ }\href {\doibase 10.1103/PhysRevB.100.115142} {\bibfield
      {journal} {\bibinfo  {journal} {Phys. Rev. B}\ }\textbf {\bibinfo {volume}
      {100}},\ \bibinfo {pages} {115142} (\bibinfo {year} {2019})}\BibitemShut
      {NoStop}%
    \bibitem [{\citenamefont {Santos}\ \emph {et~al.}(2019)\citenamefont {Santos},
      \citenamefont {\ifmmode~\mbox{\c{C}}\else \c{C}\fi{}akmak}, \citenamefont
      {Campbell},\ and\ \citenamefont {Zinner}}]{PhysRevE.100.032107}%
      \BibitemOpen
      \bibfield  {author} {\bibinfo {author} {\bibfnamefont {A.~C.}\ \bibnamefont
      {Santos}}, \bibinfo {author} {\bibfnamefont {B.~i. e. i. f. m.~c.}\
      \bibnamefont {\ifmmode~\mbox{\c{C}}\else \c{C}\fi{}akmak}}, \bibinfo {author}
      {\bibfnamefont {S.}~\bibnamefont {Campbell}}, \ and\ \bibinfo {author}
      {\bibfnamefont {N.~T.}\ \bibnamefont {Zinner}},\ }\href {\doibase
      10.1103/PhysRevE.100.032107} {\bibfield  {journal} {\bibinfo  {journal}
      {Phys. Rev. E}\ }\textbf {\bibinfo {volume} {100}},\ \bibinfo {pages}
      {032107} (\bibinfo {year} {2019})}\BibitemShut {NoStop}%
    \bibitem [{\citenamefont {Rosa}\ \emph {et~al.}(2020)\citenamefont {Rosa},
      \citenamefont {Rossini}, \citenamefont {Andolina}, \citenamefont {Polini},\
      and\ \citenamefont {Carrega}}]{Rosa_2020}%
      \BibitemOpen
      \bibfield  {author} {\bibinfo {author} {\bibfnamefont {D.}~\bibnamefont
      {Rosa}}, \bibinfo {author} {\bibfnamefont {D.}~\bibnamefont {Rossini}},
      \bibinfo {author} {\bibfnamefont {G.~M.}\ \bibnamefont {Andolina}}, \bibinfo
      {author} {\bibfnamefont {M.}~\bibnamefont {Polini}}, \ and\ \bibinfo {author}
      {\bibfnamefont {M.}~\bibnamefont {Carrega}},\ }\href {\doibase
      10.1007/jhep11(2020)067} {\bibfield  {journal} {\bibinfo  {journal} {Journal
      of High Energy Physics}\ }\textbf {\bibinfo {volume} {2020}} (\bibinfo {year}
      {2020}),\ 10.1007/jhep11(2020)067}\BibitemShut {NoStop}%
    \bibitem [{\citenamefont {Andolina}\ \emph
      {et~al.}(2019{\natexlab{a}})\citenamefont {Andolina}, \citenamefont {Keck},
      \citenamefont {Mari}, \citenamefont {Campisi}, \citenamefont {Giovannetti},\
      and\ \citenamefont {Polini}}]{PhysRevLett.122.047702}%
      \BibitemOpen
      \bibfield  {author} {\bibinfo {author} {\bibfnamefont {G.~M.}\ \bibnamefont
      {Andolina}}, \bibinfo {author} {\bibfnamefont {M.}~\bibnamefont {Keck}},
      \bibinfo {author} {\bibfnamefont {A.}~\bibnamefont {Mari}}, \bibinfo {author}
      {\bibfnamefont {M.}~\bibnamefont {Campisi}}, \bibinfo {author} {\bibfnamefont
      {V.}~\bibnamefont {Giovannetti}}, \ and\ \bibinfo {author} {\bibfnamefont
      {M.}~\bibnamefont {Polini}},\ }\href {\doibase
      10.1103/PhysRevLett.122.047702} {\bibfield  {journal} {\bibinfo  {journal}
      {Phys. Rev. Lett.}\ }\textbf {\bibinfo {volume} {122}},\ \bibinfo {pages}
      {047702} (\bibinfo {year} {2019}{\natexlab{a}})}\BibitemShut {NoStop}%
    \bibitem [{\citenamefont {Barra}(2019)}]{PhysRevLett.122.210601}%
      \BibitemOpen
      \bibfield  {author} {\bibinfo {author} {\bibfnamefont {F.}~\bibnamefont
      {Barra}},\ }\href {\doibase 10.1103/PhysRevLett.122.210601} {\bibfield
      {journal} {\bibinfo  {journal} {Phys. Rev. Lett.}\ }\textbf {\bibinfo
      {volume} {122}},\ \bibinfo {pages} {210601} (\bibinfo {year}
      {2019})}\BibitemShut {NoStop}%
    \bibitem [{\citenamefont {Hovhannisyan}\ \emph {et~al.}(2020)\citenamefont
      {Hovhannisyan}, \citenamefont {Barra},\ and\ \citenamefont
      {Imparato}}]{PhysRevResearch.2.033413}%
      \BibitemOpen
      \bibfield  {author} {\bibinfo {author} {\bibfnamefont {K.~V.}\ \bibnamefont
      {Hovhannisyan}}, \bibinfo {author} {\bibfnamefont {F.}~\bibnamefont {Barra}},
      \ and\ \bibinfo {author} {\bibfnamefont {A.}~\bibnamefont {Imparato}},\
      }\href {\doibase 10.1103/PhysRevResearch.2.033413} {\bibfield  {journal}
      {\bibinfo  {journal} {Phys. Rev. Research}\ }\textbf {\bibinfo {volume}
      {2}},\ \bibinfo {pages} {033413} (\bibinfo {year} {2020})}\BibitemShut
      {NoStop}%
    \bibitem [{\citenamefont {Allahverdyan}\ \emph {et~al.}(2004)\citenamefont
      {Allahverdyan}, \citenamefont {Balian},\ and\ \citenamefont
      {Nieuwenhuizen}}]{Allahverdyan_2004}%
      \BibitemOpen
      \bibfield  {author} {\bibinfo {author} {\bibfnamefont {A.~E.}\ \bibnamefont
      {Allahverdyan}}, \bibinfo {author} {\bibfnamefont {R.}~\bibnamefont
      {Balian}}, \ and\ \bibinfo {author} {\bibfnamefont {T.~M.}\ \bibnamefont
      {Nieuwenhuizen}},\ }\href {\doibase 10.1209/epl/i2004-10101-2} {\bibfield
      {journal} {\bibinfo  {journal} {Europhysics Letters (EPL)}\ }\textbf
      {\bibinfo {volume} {67}},\ \bibinfo {pages} {565–571} (\bibinfo {year}
      {2004})}\BibitemShut {NoStop}%
    \bibitem [{\citenamefont {Binder}\ \emph {et~al.}(2015)\citenamefont {Binder},
      \citenamefont {Vinjanampathy}, \citenamefont {Modi},\ and\ \citenamefont
      {Goold}}]{Binder_2015}%
      \BibitemOpen
      \bibfield  {author} {\bibinfo {author} {\bibfnamefont {F.~C.}\ \bibnamefont
      {Binder}}, \bibinfo {author} {\bibfnamefont {S.}~\bibnamefont
      {Vinjanampathy}}, \bibinfo {author} {\bibfnamefont {K.}~\bibnamefont {Modi}},
      \ and\ \bibinfo {author} {\bibfnamefont {J.}~\bibnamefont {Goold}},\ }\href
      {\doibase 10.1088/1367-2630/17/7/075015} {\bibfield  {journal} {\bibinfo
      {journal} {New Journal of Physics}\ }\textbf {\bibinfo {volume} {17}},\
      \bibinfo {pages} {075015} (\bibinfo {year} {2015})}\BibitemShut {NoStop}%
    \bibitem [{\citenamefont {Campaioli}\ \emph {et~al.}(2017)\citenamefont
      {Campaioli}, \citenamefont {Pollock}, \citenamefont {Binder}, \citenamefont
      {Céleri}, \citenamefont {Goold}, \citenamefont {Vinjanampathy},\ and\
      \citenamefont {Modi}}]{Campaioli_2017}%
      \BibitemOpen
      \bibfield  {author} {\bibinfo {author} {\bibfnamefont {F.}~\bibnamefont
      {Campaioli}}, \bibinfo {author} {\bibfnamefont {F.~A.}\ \bibnamefont
      {Pollock}}, \bibinfo {author} {\bibfnamefont {F.~C.}\ \bibnamefont {Binder}},
      \bibinfo {author} {\bibfnamefont {L.}~\bibnamefont {Céleri}}, \bibinfo
      {author} {\bibfnamefont {J.}~\bibnamefont {Goold}}, \bibinfo {author}
      {\bibfnamefont {S.}~\bibnamefont {Vinjanampathy}}, \ and\ \bibinfo {author}
      {\bibfnamefont {K.}~\bibnamefont {Modi}},\ }\href {\doibase
      10.1103/physrevlett.118.150601} {\bibfield  {journal} {\bibinfo  {journal}
      {Physical Review Letters}\ }\textbf {\bibinfo {volume} {118}} (\bibinfo
      {year} {2017}),\ 10.1103/physrevlett.118.150601}\BibitemShut {NoStop}%
    \bibitem [{\citenamefont {Le}\ \emph {et~al.}(2018)\citenamefont {Le},
      \citenamefont {Levinsen}, \citenamefont {Modi}, \citenamefont {Parish},\ and\
      \citenamefont {Pollock}}]{Le_2018}%
      \BibitemOpen
      \bibfield  {author} {\bibinfo {author} {\bibfnamefont {T.~P.}\ \bibnamefont
      {Le}}, \bibinfo {author} {\bibfnamefont {J.}~\bibnamefont {Levinsen}},
      \bibinfo {author} {\bibfnamefont {K.}~\bibnamefont {Modi}}, \bibinfo {author}
      {\bibfnamefont {M.~M.}\ \bibnamefont {Parish}}, \ and\ \bibinfo {author}
      {\bibfnamefont {F.~A.}\ \bibnamefont {Pollock}},\ }\href {\doibase
      10.1103/physreva.97.022106} {\bibfield  {journal} {\bibinfo  {journal}
      {Physical Review A}\ }\textbf {\bibinfo {volume} {97}} (\bibinfo {year}
      {2018}),\ 10.1103/physreva.97.022106}\BibitemShut {NoStop}%
    \bibitem [{\citenamefont {Ferraro}\ \emph {et~al.}(2018)\citenamefont
      {Ferraro}, \citenamefont {Campisi}, \citenamefont {Andolina}, \citenamefont
      {Pellegrini},\ and\ \citenamefont {Polini}}]{PhysRevLett.120.117702}%
      \BibitemOpen
      \bibfield  {author} {\bibinfo {author} {\bibfnamefont {D.}~\bibnamefont
      {Ferraro}}, \bibinfo {author} {\bibfnamefont {M.}~\bibnamefont {Campisi}},
      \bibinfo {author} {\bibfnamefont {G.~M.}\ \bibnamefont {Andolina}}, \bibinfo
      {author} {\bibfnamefont {V.}~\bibnamefont {Pellegrini}}, \ and\ \bibinfo
      {author} {\bibfnamefont {M.}~\bibnamefont {Polini}},\ }\href {\doibase
      10.1103/PhysRevLett.120.117702} {\bibfield  {journal} {\bibinfo  {journal}
      {Phys. Rev. Lett.}\ }\textbf {\bibinfo {volume} {120}},\ \bibinfo {pages}
      {117702} (\bibinfo {year} {2018})}\BibitemShut {NoStop}%
    \bibitem [{\citenamefont {Andolina}\ \emph
      {et~al.}(2019{\natexlab{b}})\citenamefont {Andolina}, \citenamefont {Keck},
      \citenamefont {Mari}, \citenamefont {Giovannetti},\ and\ \citenamefont
      {Polini}}]{PhysRevB.99.205437}%
      \BibitemOpen
      \bibfield  {author} {\bibinfo {author} {\bibfnamefont {G.~M.}\ \bibnamefont
      {Andolina}}, \bibinfo {author} {\bibfnamefont {M.}~\bibnamefont {Keck}},
      \bibinfo {author} {\bibfnamefont {A.}~\bibnamefont {Mari}}, \bibinfo {author}
      {\bibfnamefont {V.}~\bibnamefont {Giovannetti}}, \ and\ \bibinfo {author}
      {\bibfnamefont {M.}~\bibnamefont {Polini}},\ }\href {\doibase
      10.1103/PhysRevB.99.205437} {\bibfield  {journal} {\bibinfo  {journal} {Phys.
      Rev. B}\ }\textbf {\bibinfo {volume} {99}},\ \bibinfo {pages} {205437}
      (\bibinfo {year} {2019}{\natexlab{b}})}\BibitemShut {NoStop}%
    \bibitem [{\citenamefont {Crescente}\ \emph
      {et~al.}(2020{\natexlab{b}})\citenamefont {Crescente}, \citenamefont
      {Carrega}, \citenamefont {Sassetti},\ and\ \citenamefont
      {Ferraro}}]{Crescente_2020}%
      \BibitemOpen
      \bibfield  {author} {\bibinfo {author} {\bibfnamefont {A.}~\bibnamefont
      {Crescente}}, \bibinfo {author} {\bibfnamefont {M.}~\bibnamefont {Carrega}},
      \bibinfo {author} {\bibfnamefont {M.}~\bibnamefont {Sassetti}}, \ and\
      \bibinfo {author} {\bibfnamefont {D.}~\bibnamefont {Ferraro}},\ }\href
      {\doibase 10.1103/physrevb.102.245407} {\bibfield  {journal} {\bibinfo
      {journal} {Physical Review B}\ }\textbf {\bibinfo {volume} {102}} (\bibinfo
      {year} {2020}{\natexlab{b}}),\ 10.1103/physrevb.102.245407}\BibitemShut
      {NoStop}%
    \bibitem [{\citenamefont {Ghosh}\ \emph
      {et~al.}(2020{\natexlab{a}})\citenamefont {Ghosh}, \citenamefont {Chanda},\
      and\ \citenamefont {Sen(De)}}]{PhysRevA.101.032115}%
      \BibitemOpen
      \bibfield  {author} {\bibinfo {author} {\bibfnamefont {S.}~\bibnamefont
      {Ghosh}}, \bibinfo {author} {\bibfnamefont {T.}~\bibnamefont {Chanda}}, \
      and\ \bibinfo {author} {\bibfnamefont {A.}~\bibnamefont {Sen(De)}},\ }\href
      {\doibase 10.1103/PhysRevA.101.032115} {\bibfield  {journal} {\bibinfo
      {journal} {Phys. Rev. A}\ }\textbf {\bibinfo {volume} {101}},\ \bibinfo
      {pages} {032115} (\bibinfo {year} {2020}{\natexlab{a}})}\BibitemShut
      {NoStop}%
    \bibitem [{\citenamefont {Rossini}\ \emph {et~al.}(2020)\citenamefont
      {Rossini}, \citenamefont {Andolina}, \citenamefont {Rosa}, \citenamefont
      {Carrega},\ and\ \citenamefont {Polini}}]{Rossini_2020}%
      \BibitemOpen
      \bibfield  {author} {\bibinfo {author} {\bibfnamefont {D.}~\bibnamefont
      {Rossini}}, \bibinfo {author} {\bibfnamefont {G.~M.}\ \bibnamefont
      {Andolina}}, \bibinfo {author} {\bibfnamefont {D.}~\bibnamefont {Rosa}},
      \bibinfo {author} {\bibfnamefont {M.}~\bibnamefont {Carrega}}, \ and\
      \bibinfo {author} {\bibfnamefont {M.}~\bibnamefont {Polini}},\ }\href
      {\doibase 10.1103/physrevlett.125.236402} {\bibfield  {journal} {\bibinfo
      {journal} {Physical Review Letters}\ }\textbf {\bibinfo {volume} {125}}
      (\bibinfo {year} {2020}),\ 10.1103/physrevlett.125.236402}\BibitemShut
      {NoStop}%
    \bibitem [{\citenamefont {Zakavati}\ \emph {et~al.}(2020)\citenamefont
      {Zakavati}, \citenamefont {Tabesh},\ and\ \citenamefont
      {Salimi}}]{zakavati2020bounds}%
      \BibitemOpen
      \bibfield  {author} {\bibinfo {author} {\bibfnamefont {S.}~\bibnamefont
      {Zakavati}}, \bibinfo {author} {\bibfnamefont {F.~T.}\ \bibnamefont
      {Tabesh}}, \ and\ \bibinfo {author} {\bibfnamefont {S.}~\bibnamefont
      {Salimi}},\ }\href@noop {} {\enquote {\bibinfo {title} {Bounds on charging
      power of open quantum batteries},}\ } (\bibinfo {year} {2020}),\ \Eprint
      {http://arxiv.org/abs/2003.09814} {arXiv:2003.09814 [quant-ph]} \BibitemShut
      {NoStop}%
    \bibitem [{\citenamefont {Ghosh}\ \emph
      {et~al.}(2020{\natexlab{b}})\citenamefont {Ghosh}, \citenamefont {Chanda},
      \citenamefont {Mal},\ and\ \citenamefont {De}}]{ghosh2020fast}%
      \BibitemOpen
      \bibfield  {author} {\bibinfo {author} {\bibfnamefont {S.}~\bibnamefont
      {Ghosh}}, \bibinfo {author} {\bibfnamefont {T.}~\bibnamefont {Chanda}},
      \bibinfo {author} {\bibfnamefont {S.}~\bibnamefont {Mal}}, \ and\ \bibinfo
      {author} {\bibfnamefont {A.~S.}\ \bibnamefont {De}},\ }\href@noop {}
      {\enquote {\bibinfo {title} {Fast charging of quantum battery assisted by
      noise},}\ } (\bibinfo {year} {2020}{\natexlab{b}}),\ \Eprint
      {http://arxiv.org/abs/2005.12859} {arXiv:2005.12859 [quant-ph]} \BibitemShut
      {NoStop}%
    \bibitem [{\citenamefont {Seah}\ \emph {et~al.}(2021)\citenamefont {Seah},
      \citenamefont {Perarnau-Llobet}, \citenamefont {Haack}, \citenamefont
      {Brunner},\ and\ \citenamefont {Nimmrichter}}]{seah2021quantum}%
      \BibitemOpen
      \bibfield  {author} {\bibinfo {author} {\bibfnamefont {S.}~\bibnamefont
      {Seah}}, \bibinfo {author} {\bibfnamefont {M.}~\bibnamefont
      {Perarnau-Llobet}}, \bibinfo {author} {\bibfnamefont {G.}~\bibnamefont
      {Haack}}, \bibinfo {author} {\bibfnamefont {N.}~\bibnamefont {Brunner}}, \
      and\ \bibinfo {author} {\bibfnamefont {S.}~\bibnamefont {Nimmrichter}},\
      }\href@noop {} {\enquote {\bibinfo {title} {Quantum speed-up in collisional
      battery charging},}\ } (\bibinfo {year} {2021}),\ \Eprint
      {http://arxiv.org/abs/2105.01863} {arXiv:2105.01863 [quant-ph]} \BibitemShut
      {NoStop}%
    \bibitem [{\citenamefont {Garc\'{\i}a-Pintos}\ \emph
      {et~al.}(2020)\citenamefont {Garc\'{\i}a-Pintos}, \citenamefont {Hamma},\
      and\ \citenamefont {del Campo}}]{PhysRevLett.125.040601}%
      \BibitemOpen
      \bibfield  {author} {\bibinfo {author} {\bibfnamefont {L.~P.}\ \bibnamefont
      {Garc\'{\i}a-Pintos}}, \bibinfo {author} {\bibfnamefont {A.}~\bibnamefont
      {Hamma}}, \ and\ \bibinfo {author} {\bibfnamefont {A.}~\bibnamefont {del
      Campo}},\ }\href {\doibase 10.1103/PhysRevLett.125.040601} {\bibfield
      {journal} {\bibinfo  {journal} {Phys. Rev. Lett.}\ }\textbf {\bibinfo
      {volume} {125}},\ \bibinfo {pages} {040601} (\bibinfo {year}
      {2020})}\BibitemShut {NoStop}%
    \bibitem [{\citenamefont {Zhang}\ and\ \citenamefont
      {blaauboer}(2018)}]{zhang2018enhanced}%
      \BibitemOpen
      \bibfield  {author} {\bibinfo {author} {\bibfnamefont {X.}~\bibnamefont
      {Zhang}}\ and\ \bibinfo {author} {\bibfnamefont {M.}~\bibnamefont
      {blaauboer}},\ }\href@noop {} {\enquote {\bibinfo {title} {Enhanced energy
      transfer in a dicke quantum battery},}\ } (\bibinfo {year} {2018}),\ \Eprint
      {http://arxiv.org/abs/1812.10139} {arXiv:1812.10139 [quant-ph]} \BibitemShut
      {NoStop}%
    \bibitem [{\citenamefont {Julià-Farré}\ \emph {et~al.}(2020)\citenamefont
      {Julià-Farré}, \citenamefont {Salamon}, \citenamefont {Riera},
      \citenamefont {Bera},\ and\ \citenamefont {Lewenstein}}]{Juli_Farr__2020}%
      \BibitemOpen
      \bibfield  {author} {\bibinfo {author} {\bibfnamefont {S.}~\bibnamefont
      {Julià-Farré}}, \bibinfo {author} {\bibfnamefont {T.}~\bibnamefont
      {Salamon}}, \bibinfo {author} {\bibfnamefont {A.}~\bibnamefont {Riera}},
      \bibinfo {author} {\bibfnamefont {M.~N.}\ \bibnamefont {Bera}}, \ and\
      \bibinfo {author} {\bibfnamefont {M.}~\bibnamefont {Lewenstein}},\ }\href
      {\doibase 10.1103/physrevresearch.2.023113} {\bibfield  {journal} {\bibinfo
      {journal} {Physical Review Research}\ }\textbf {\bibinfo {volume} {2}}
      (\bibinfo {year} {2020}),\ 10.1103/physrevresearch.2.023113}\BibitemShut
      {NoStop}%
    \bibitem [{Note1()}]{Note1}%
      \BibitemOpen
      \bibinfo {note} {It is proved by inserting Eq.~\protect \textup {\hbox
      {\mathsurround \z@ \protect \normalfont (\ignorespaces \ref
      {eq:state_evolution}\unskip \@@italiccorr )}} to Eq.~\protect \textup {\hbox
      {\mathsurround \z@ \protect \normalfont (\ignorespaces \ref {eq:power}\unskip
      \@@italiccorr )}}, which gives $P(t)=\protect \mathrm {tr}({\protect \hat
      {\rho }}[\protect \hat {H},\protect \hat {V}])=\DOTSB \tsum \slimits@ _k\rho
      _k\langle \psi _k|[\protect \hat {H},\protect \hat {V}]|\psi _k\rangle \leq
      \left \|[\protect \hat {H},\protect \hat {V}]\right \|\leq 2\left \|\protect
      \hat {H}\right \|\left \|\protect \hat {V}\right \|$, using the spectral
      decomposition ${\protect \hat {\rho }}=\DOTSB \tsum \slimits@ _k\rho _k|\psi
      _k\rangle \langle \psi _k|$ and $\DOTSB \tsum \slimits@ _k\rho _k=1$. See,
      for example, Ref.~\cite {Juli_Farr__2020}}\BibitemShut {NoStop}%
    \bibitem [{\citenamefont {Sachdev}\ and\ \citenamefont
      {Ye}(1993)}]{PhysRevLett.70.3339}%
      \BibitemOpen
      \bibfield  {author} {\bibinfo {author} {\bibfnamefont {S.}~\bibnamefont
      {Sachdev}}\ and\ \bibinfo {author} {\bibfnamefont {J.}~\bibnamefont {Ye}},\
      }\href {\doibase 10.1103/PhysRevLett.70.3339} {\bibfield  {journal} {\bibinfo
       {journal} {Phys. Rev. Lett.}\ }\textbf {\bibinfo {volume} {70}},\ \bibinfo
      {pages} {3339} (\bibinfo {year} {1993})}\BibitemShut {NoStop}%
    \bibitem [{\citenamefont {Bray}\ and\ \citenamefont {Moore}(1980)}]{Bray_1980}%
      \BibitemOpen
      \bibfield  {author} {\bibinfo {author} {\bibfnamefont {A.~J.}\ \bibnamefont
      {Bray}}\ and\ \bibinfo {author} {\bibfnamefont {M.~A.}\ \bibnamefont
      {Moore}},\ }\href {\doibase 10.1088/0022-3719/13/24/005} {\bibfield
      {journal} {\bibinfo  {journal} {Journal of Physics C: Solid State Physics}\
      }\textbf {\bibinfo {volume} {13}},\ \bibinfo {pages} {L655} (\bibinfo {year}
      {1980})}\BibitemShut {NoStop}%
    \bibitem [{\citenamefont {{Giovannetti}}\ \emph {et~al.}(2004)\citenamefont
      {{Giovannetti}}, \citenamefont {{Lloyd}},\ and\ \citenamefont
      {{Maccone}}}]{giovannetti2004quantum}%
      \BibitemOpen
      \bibfield  {author} {\bibinfo {author} {\bibfnamefont {V.}~\bibnamefont
      {{Giovannetti}}}, \bibinfo {author} {\bibfnamefont {S.}~\bibnamefont
      {{Lloyd}}}, \ and\ \bibinfo {author} {\bibfnamefont {L.}~\bibnamefont
      {{Maccone}}},\ }\href {\doibase 10.1126/science.1104149} {\bibfield
      {journal} {\bibinfo  {journal} {Science}\ }\textbf {\bibinfo {volume}
      {306}},\ \bibinfo {pages} {1330} (\bibinfo {year} {2004})},\ \Eprint
      {http://arxiv.org/abs/quant-ph/0412078} {arXiv:quant-ph/0412078 [quant-ph]}
      \BibitemShut {NoStop}%
    \bibitem [{\citenamefont {{Demkowicz-Dobrza{\'n}ski}}\ \emph
      {et~al.}(2012)\citenamefont {{Demkowicz-Dobrza{\'n}ski}}, \citenamefont
      {{Ko{\l}ody{\'n}ski}},\ and\ \citenamefont
      {{Gu{\c{t}}{\u{a}}}}}]{demkowicz2012elusive}%
      \BibitemOpen
      \bibfield  {author} {\bibinfo {author} {\bibfnamefont {R.}~\bibnamefont
      {{Demkowicz-Dobrza{\'n}ski}}}, \bibinfo {author} {\bibfnamefont
      {J.}~\bibnamefont {{Ko{\l}ody{\'n}ski}}}, \ and\ \bibinfo {author}
      {\bibfnamefont {M.}~\bibnamefont {{Gu{\c{t}}{\u{a}}}}},\ }\href {\doibase
      10.1038/ncomms2067} {\bibfield  {journal} {\bibinfo  {journal} {Nature
      Communications}\ }\textbf {\bibinfo {volume} {3}},\ \bibinfo {eid} {1063}
      (\bibinfo {year} {2012})},\ \Eprint {http://arxiv.org/abs/1201.3940}
      {arXiv:1201.3940 [quant-ph]} \BibitemShut {NoStop}%
    \bibitem [{\citenamefont {Grover}(1996)}]{grover1996fast}%
      \BibitemOpen
      \bibfield  {author} {\bibinfo {author} {\bibfnamefont {L.~K.}\ \bibnamefont
      {Grover}},\ }in\ \href {https://dl.acm.org/doi/10.1145/237814.237866} {\emph
      {\bibinfo {booktitle} {Proceedings of the twenty-eighth annual ACM symposium
      on Theory of computing}}}\ (\bibinfo {year} {1996})\ pp.\ \bibinfo {pages}
      {212--219}\BibitemShut {NoStop}%
    \bibitem [{\citenamefont {Bennett}\ \emph {et~al.}(1997)\citenamefont
      {Bennett}, \citenamefont {Bernstein}, \citenamefont {Brassard},\ and\
      \citenamefont {Vazirani}}]{bennett1997strengths}%
      \BibitemOpen
      \bibfield  {author} {\bibinfo {author} {\bibfnamefont {C.~H.}\ \bibnamefont
      {Bennett}}, \bibinfo {author} {\bibfnamefont {E.}~\bibnamefont {Bernstein}},
      \bibinfo {author} {\bibfnamefont {G.}~\bibnamefont {Brassard}}, \ and\
      \bibinfo {author} {\bibfnamefont {U.}~\bibnamefont {Vazirani}},\ }\href
      {https://epubs.siam.org/doi/10.1137/S0097539796300933} {\bibfield  {journal}
      {\bibinfo  {journal} {SIAM journal on Computing}\ }\textbf {\bibinfo {volume}
      {26}},\ \bibinfo {pages} {1510} (\bibinfo {year} {1997})}\BibitemShut
      {NoStop}%
    \end{thebibliography}

\begin{thebibliography}{1}%
    \makeatletter
    \providecommand \@ifxundefined [1]{%
     \@ifx{#1\undefined}
    }%
    \providecommand \@ifnum [1]{%
     \ifnum #1\expandafter \@firstoftwo
     \else \expandafter \@secondoftwo
     \fi
    }%
    \providecommand \@ifx [1]{%
     \ifx #1\expandafter \@firstoftwo
     \else \expandafter \@secondoftwo
     \fi
    }%
    \providecommand \natexlab [1]{#1}%
    \providecommand \enquote  [1]{``#1''}%
    \providecommand \bibnamefont  [1]{#1}%
    \providecommand \bibfnamefont [1]{#1}%
    \providecommand \citenamefont [1]{#1}%
    \providecommand \href@noop [0]{\@secondoftwo}%
    \providecommand \href [0]{\begingroup \@sanitize@url \@href}%
    \providecommand \@href[1]{\@@startlink{#1}\@@href}%
    \providecommand \@@href[1]{\endgroup#1\@@endlink}%
    \providecommand \@sanitize@url [0]{\catcode `\\12\catcode `\$12\catcode
      `\&12\catcode `\#12\catcode `\^12\catcode `\_12\catcode `\%12\relax}%
    \providecommand \@@startlink[1]{}%
    \providecommand \@@endlink[0]{}%
    \providecommand \url  [0]{\begingroup\@sanitize@url \@url }%
    \providecommand \@url [1]{\endgroup\@href {#1}{\urlprefix }}%
    \providecommand \urlprefix  [0]{URL }%
    \providecommand \Eprint [0]{\href }%
    \providecommand \doibase [0]{http://dx.doi.org/}%
    \providecommand \selectlanguage [0]{\@gobble}%
    \providecommand \bibinfo  [0]{\@secondoftwo}%
    \providecommand \bibfield  [0]{\@secondoftwo}%
    \providecommand \translation [1]{[#1]}%
    \providecommand \BibitemOpen [0]{}%
    \providecommand \bibitemStop [0]{}%
    \providecommand \bibitemNoStop [0]{.\EOS\space}%
    \providecommand \EOS [0]{\spacefactor3000\relax}%
    \providecommand \BibitemShut  [1]{\csname bibitem#1\endcsname}%
    \let\auto@bib@innerbib\@empty
    %</preamble>
    \bibitem [{\citenamefont {Kamath}(2015)}]{kamath2015bounds}%
      \BibitemOpen
      \bibfield  {author} {\bibinfo {author} {\bibfnamefont {G.}~\bibnamefont
      {Kamath}},\ }\href {http://www.gautamkamath.com/writings/gaussian max.pdf}
      {\enquote {\bibinfo {title} {Bounds on the expectation of the maximum of
      samples from a gaussian},}\ } (\bibinfo {year} {2015})\BibitemShut {NoStop}%
    \end{thebibliography}
\end{document}